\documentclass[11pt, twoside, letter]{article}
\usepackage[margin = 1in]{geometry}
\usepackage[utf8]{inputenc}
\usepackage{amsmath, amssymb, amsfonts, cite}
\usepackage{ifthen}
\usepackage[boxed, linesnumbered]{algorithm2e}
\usepackage{comment}

\usepackage{pgf}
\usepackage{tikz}
\usetikzlibrary{arrows,automata}

\usepackage{microtype}

\usepackage{amsthm}
\usepackage{hyperref}
\usepackage{multirow, multicol}
\usepackage{subfigure, graphicx} 
\usetikzlibrary{arrows,shapes}

\tikzset{circle/.pic={
		\node[circle, aspect=1, draw, minimum size=0.3cm, text width=0.2cm] () at (0,0) {\tikzpictext};
	}} 
\usepackage{natbib}
\usepackage{microtype}
\usepackage{pgfplots}
\pgfplotsset{compat=1.10}

\usepackage{bm}

\newcommand{\etal}{{\em et al.~}}
\newcommand{\vv}[1]{\bm{#1}}
\newcommand{\R}{\mathbb{R}}

\newtheorem{theorem}{Theorem}
\newtheorem{definition}{Definition}
\newtheorem{lemma}{Lemma}
\newtheorem{proposition}{Proposition}

\newtheorem{example}{Example}

\begin{document}
\title{Bifurcation Mechanism Design -- From Optimal Flat Taxes to Improved Cancer Treatments}  
\author{Ger Yang\\
University of Texas at Austin\\
Department of Electrical and Computer Engineering\\
\and
Georgios Piliouras\\
Singapore University of Technology and Design\\
Engineering Systems and Design (ESD)
\and
David Basanta\\
Integrated Mathematical Oncology\\
H. Lee Moffitt Cancer Center and Research Institute
}

\maketitle

\begin{abstract}


 Small changes to the parameters of a system can lead to abrupt qualitative changes of its behavior, a phenomenon known as bifurcation.  
Such instabilities are typically considered problematic, however, we show that their power can be leveraged to design novel types of mechanisms. 
	\textit{Hysteresis mechanisms} use transient changes of system parameters to induce a
	permanent improvement to its performance via optimal equilibrium selection. 
	\textit{Optimal control mechanisms} induce convergence to states whose performance is better than even the best equilibrium.
We apply these mechanisms in two different settings that illustrate the versatility of bifurcation mechanism design. In the first one we explore how introducing flat taxation can improve social welfare, despite decreasing agent ``rationality", by destabilizing inefficient equilibria. From there we move on to consider a well known game of tumor metabolism and use our approach to derive novel cancer treatment strategies.

\end{abstract}


\thanks{This work is supported by the National Science Foundation,
  under grant CNS-0435060, grant CCR-0325197 and grant EN-CS-0329609.}


\section{Introduction}

The term bifurcation, which means splitting in two, is used to describe abrupt qualitative
changes in system behavior due to smooth variation of its parameters. 
Bifurcations are ubiquitous 
and  permeate all natural phenomena. 
Effectively, they 
 produce discrete events (e.g., rain breaking out) out of
 smoothly varying, continuous systems (e.g., small changes to humidity, temperature).
Typically, they are studied through bifurcation diagrams, multi-valued maps that prescribe how each 
parameter configuration  translates to possible system behaviors (e.g., Figure~\ref{fig:intro}).   


Bifurcations arise in a natural way in game theory. Games are typically studied through their Nash correspondences,
a multi-valued map connecting the parameters of the game (i.e., payoff matrices) to system behavior, in this case Nash equilibria.
As we slowly vary the parameters of the game typically the Nash equilibria will also vary smoothly, except at bifurcation points
where, for example, the number of equilibria abruptly changes as some equilibria appear/disappear altogether. 
Such singularities may a have huge impact both on system behavior and system performance. For example, if the system state was at an equilibrium that disappeared during the bifurcation, then a  turbulent transitionary period ensues where the system tries to reorganize itself at one of the remaining equilibria. Moreover, the quality of all remaining equilibria may be significantly worse 
than the original.  Even more disturbingly, it is not a-priori clear that the system will equilibrate at all. 
Successive bifurcations that lead to increasingly more complicated recurrent behavior is a standard route to chaos \cite{devaney1992first}, which may have devastating effects to system performance.

Game theorists are particularly aware of the need to produce ``robust" predictions that are not inherently bound
 to a specific, exact instantiation of the payoff parameters of the game \cite{roughgarden09}. The typical way to approach this problem has been to focus on more
expansive solution concepts, e.g., $\epsilon$-approximate Nash equilibria or even outcomes approximately consistent to regret-minimizing learning. These approaches, however, do not really address
the problem at its core as any solution concept defines a map from parameter space to behavioral space and no such map is  immune to bifurcations. If pushed hard enough any system will destabilize. The question is what happens next?

Well, a lot of things \textit{may} happen. It is intuitively clear that if we are allowed to play around arbitrarily with the payoffs of
the agents then we can reproduce any game and no meaningful analysis is possible. Using payoff entries 
as controlling parameters is problematic for another reason. It is not clear that there exists a compelling 
 parametrization of the payoff space that captures how real life decision makers deviate from the Platonic ideal of the payoff matrix.
Instead, we focus on another popular aspect of economic theory, agent ``rationality".

We adopt a standard model of boundedly rational learning agents.
Boltzmann Q-learning dynamics \cite{watkins1989learning,watkins1992q,tan1993multi} 
is a well studied behavioral model in which agents are parameterized by a temperature/rationality term $T$.
Each agent
keeps track of the collective past performance of his actions (i.e., learns from experience)
 and chooses an action according to a Boltzmann/Gibbs distribution with parameter $T$.
 When applied to a multi-agent game the behavioral fixed points of  Q-learning are known as quantal response equilibria (QRE) \cite{McKelvey:1995aa}. 
 Naturally, QREs depend on the temperature $T$. As $T\rightarrow0$  players become perfectly rational and play approaches a Nash equilibrium,\footnote{Mixed strategies in the QRE model are sometimes interpreted as frequency distributions of deterministic actions in a large population of users. This population interpretation of mixed strategies is standard and dates back to Nash~\cite{Nash}. Depending on context, we will use either the probabilistic interpretation or the population one.} whereas as $T\rightarrow \infty$ all agents use uniformly random strategies. As we vary the temperature the QRE($T$) correspondence moves between these two extremes  producing bifurcations along the way at critical points where the number of QREs changes (Figure~\ref{fig:intro}). 
 
 \medskip
 
 \textit{Our goal} in this paper is to quantify the effects of these rationality-driven bifurcations to the social welfare of two player two strategy games. At this point a moment of pause is warranted. Why is this a worthy goal? Games of small size  ($2 \times 2$ games in particular) hardly seem like a subject worthy of serious scientific investigation. This, however, could not be further from the truth. 
 
 First, the correct way to interpret this setting is from the point of population games where each agent is better understood as a large homogeneous population (e.g. men and women, attackers and defenders, cells of type A and cells of type B). Each of a handful of different types of users
 has only a few meaningful actions available to them. In fact, from the perspective of applied game theory only such games with a small number of parameters are practically meaningful. The reason should be clear by now. Any game theoretic modeling of a real life scenario is  invariably noisy and inaccurate.  In order for game-theoretic predictions to be practically binding they have to be robust to these uncertainties. If the system intrinsically has a large number of independent parameters e.g., 20, then this parameter space will almost certainly encode a vast number of bifurcations, which invalidate any theoretical prediction.  Practically useful models \textit{need} to be small. 
 
 Secondly, game theoretic models applied  for scientific purposes typically \textit{are} small. 
 Specifically, the exact setting studied here with Boltzmann Q-learning dynamics applied in $2 \times 2$ games has been used to model the effects of taxation to agent rationality \cite{wolpert2012hysteresis} (see Section~\ref{sec:taxation} for a more extensive discussion) as well as to model the effects of treatments that trigger phase transitions to cancer dynamics \cite{kianercy2014critical} (see Section~\ref{sec:cancer}). Our approach yields insights to explicit open questions in both of these applications areas. In fact, direct application of our analysis can address similar inquiries for any other phenomenon modeled by  Q-learning dynamics applied in $2 \times 2$ games. 
 
Finally, the analysis itself is far from straightforward as it requires combining sets of tools and techniques that have so far been developed in isolation from each other. On one hand, we need to understand the behavior of these dynamical systems using tools from topology of dynamical systems whose implications are largely qualitative (e.g. prove the lack of cyclic trajectories). On the other hand, we need to leverage these tools to quantify at which exact parameter values bifurcations occur and 
 produce price-of-anarchy type of guarantees which by definition are  quantitative.  As far as we know, this is the first instance of a fruitful combination of these tools.  In fact, not only do we show how to analyze the effects of bifurcations to system efficiency, we also show how to leverage this understanding (e.g. knowledge of the geometry of the bifurcation diagrams) to design novel types of mechanisms with good performance guarantees.

 
 


\begin{figure}[t]
  \centering
  \includegraphics[width=0.3\linewidth]{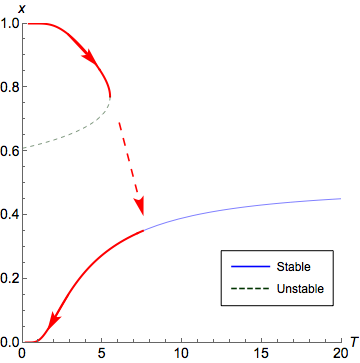}
  \caption{Bifurcation diagram for a $2 \times 2$ population coordination game. The $x$ axis corresponds to the system temperature $T$, whereas
  the $y$ axis corresponds to the projection of the proportion of the first population using the first strategy at equilibrium. For small $T$, the system exhibits multiple equilibria.
   Starting at $T=0$, and  by increasing the temperature beyond
the critical threshold $T_C=6$, and then bringing it back to zero, we can force the system to converge to another equilibrium.}
  \label{fig:intro}
\end{figure}

\vspace{-3pt}
\subsection*{Our contribution.} 

We introduce two different types of mechanisms, hysteresis and optimal control mechanisms.

\textit{Hysteresis mechanisms} use transient changes to the system parameters to induce permanent improvements to its performance via optimal (Nash) equilibrium selection. The term hysteresis is derived from an ancient Greek word that means ``to lag behind". It reflects a time-based dependence between the system's present output and its past inputs. For example, let's assume that we start from a game theoretic system of Q-learning agents with temperature $T=0$ and assume that the system has converged to an equilibrium. By increasing the temperature beyond some critical threshold and then bringing it back to zero, we can force the system to provably converge to another equilibrium, e.g., the best (Nash) equilibrium (Figure~\ref{fig:intro}, Theorem~\ref{thm:mec2}). Thus, we can ensure performance equivalent to that of the price of stability instead of the price of anarchy. One attractive feature of this mechanism is that from the perspective of the central designer it is rather ``cheap" to implement. Whereas typical mechanisms require the designer to continuously intervene by (e.g., by paying the agents) to offset their greedy tendencies this mechanism is transient with a finite amount of total effort from the perspective of the designer. Further, the idea that game theoretic systems have effectively systemic memory is rather interesting and could find other applications  within algorithmic game theory.

\textit{Optimal control mechanisms}  induce convergence to states whose performance is better than even the best Nash equilibrium. 
Thus, we can at times even beat the price of stability (Theorem \ref{thm:mec1}). Specifically,  we show that by controlling the exploration/exploitation tradeoff we can achieve strictly better states than those achievable by perfectly rational agents. In order to implement such a mechanism it does not suffice to identify the right set of agents' parameters/temperatures so that the system has some QRE whose social welfare is better than the best Nash. We need to design a trajectory through the parameter space so that this optimal QRE becomes  the final resting point.

\section{Preliminaries}
\label{sec:prelim}

\subsection{Game Theory Basics: $2 \times 2$ games} \label{sec:def_2_by_2_game}
In this paper, we focus on $2 \times 2$ games.  We define it as a game with two players, and each player has two actions.  We write the payoff matrices of the game for each player as
\begin{equation} \label{eq:def_AB}
\vv{A} = \left( \begin{array}{cc}
a_{11} & a_{12} \\
a_{21} & a_{22} \end{array}
\right), \quad
\vv{B} = \left( \begin{array}{cc}
b_{11} & b_{12} \\
b_{21} & b_{22} \end{array}
\right)
\end{equation}
respectively.  The entry $a_{ij}$ denotes the payoff for Player~$1$ when he chooses action $i$ and his opponent chooses action $j$; similarly, $b_{ij}$ denotes the payoff for Player~$2$ when he chooses action $i$ and his opponent chooses action $j$.  We define $x$ as the probability that the Player~$1$ chooses his first action, and $y$ as the probability that Player~$2$ chooses his first action.  We also define two row vectors $\vv{x}=(x, 1-x)^T$ and $\vv{y}=(y,1-y)^T$ as the \emph{strategy} for each player.  For simplicity, we denote the $i$-th entry of vector $\vv{x}$ by $x_i$.   We call the tuple $(x,y)$ as the \emph{system state} or the \emph{strategy profile}.  

An important solution concept in game theory is the \emph{Nash equilibrium}, where each user cannot make profit by unilaterally changing his strategy, that is:
\begin{definition}[Nash equilibrium]
	A strategy profile $(x_{NE},y_{NE})$ is a Nash equilibrium (NE) if
	\begin{align*}
		&x_{NE} \in \arg\max_{x \in [0,1]} \vv{x}^T\vv{A}\vv{y}_{NE},
		&y_{NE} \in \arg\max_{y \in [0,1]} \vv{y}^T\vv{B}\vv{x}_{NE}
	\end{align*}
\end{definition}

We call $(x_{NE},y_{NE})$ a \emph{pure} Nash equilibrium (PNE) if both $x_{NE} \in \{0,1\}$ and $y_{NE} \in \{0,1\}$.
Nash equilibrium assumes each user is fully rational.  However, in real world, this assumption is impractical.  An alternative solution concept is the \emph{quantal response equilibrium} \cite{McKelvey:1995aa}, where it assumes that each user has bounded rationality:
\begin{definition}[Quantal response equilibrium]
	A strategy profile $(x_{QRE},y_{QRE})$ is a Quantal response equilibrium (QRE) with respect to temperature $T_x$ and $T_y$ if
	\begin{align*}
		x_{QRE} &= \frac{e^{\frac{1}{T_x} (\vv{A}\vv{y}_{QRE})_1}}{\sum_{j\in\{1,2\}} e^{\frac{1}{T_x} (\vv{A}\vv{y}_{QRE})_j}}, &1-x_{QRE} = \frac{e^{\frac{1}{T_x} (\vv{A}\vv{y}_{QRE})_2}}{\sum_{j\in\{1,2\}} e^{\frac{1}{T_x} (\vv{A}\vv{y}_{QRE})_j}}  \\
		y_{QRE} &= \frac{e^{\frac{1}{T_y} (\vv{B}\vv{x}_{QRE})_1}}{\sum_{j\in\{1,2\}} e^{\frac{1}{T_y} (\vv{B}\vv{x}_{QRE})_j}}, &1-y_{QRE} = \frac{e^{\frac{1}{T_y} (\vv{B}\vv{x}_{QRE})_2}}{\sum_{j\in\{1,2\}} e^{\frac{1}{T_y} (\vv{B}\vv{x}_{QRE})_j}}  
	\end{align*}
\end{definition}

Analogous to the definition of Nash equilibria, we can consider the QREs as the case that each player is not only maximizing the expected utility but also maximizing the entropy.  We can see that the QREs are the solutions to maximizing the linear combination of the following program:
\begin{align*}
\vv{x}_{QRE} &\in \arg\max_{\vv{x}} \left\{\vv{x}^T\vv{A}\vv{y}_{QRE} - T_x \sum_{j} x_j \ln x_j \right\}\\
\vv{y}_{QRE} &\in \arg\max_{\vv{y}} \left\{ \vv{y}^T\vv{B}\vv{x}_{QRE} - T_y \sum_{j} y_j \ln y_j \right\}
\end{align*}
This formulation has been widely seen in Q-learning dynamics literature (e.g \cite{Cominetti:2010aa, wolpert2012hysteresis, coucheney2013entropy}).  With this formulation, we can find that the two parameters $T_x$ and $T_y$ controls the weighting between the utility and the entropy.  We call $T_x$ and $T_y$ the \emph{temperatures}, and their value defines the level of irrationality.  If $T_x$ and $T_y$ are zero, then both players are fully rational, and the system state is a Nash equilibrium.  However, if both $T_x$ and $T_y$ are infinity, then each player is choosing his action according to a uniform distribution, which corresponds to the fully irrational players.

\subsection{Efficiency of an equilibrium} \label{sec:eff_game_intro}
The performance of a system state can be measured via the \emph{social welfare}.  Given a system state $(x,y)$, we define the social welfare as the sum of the expected payoff of all users in the system:
\begin{definition}
	Given $2 \times 2$ game with payoff matrices $\vv{A}$ and $\vv{B}$, and a system state $(x,y)$, the social welfare is defined as
	$$
	SW(x,y)=xy(a_{11}+b_{11})+x(1-y)(a_{12}+b_{21})+y(1-x)(a_{21}+b_{12})+(1-x)(1-y)(a_{22}+b_{22})
	$$
\end{definition}

In the context of algorithmic game theory, we can measure the efficiency of a game by comparing the social welfare of a equilibrium system state with the best social welfare.  We call the strategy profile that achieves the maximal social welfare as the \emph{socially optimal (SO)} strategy profile.  The efficiency of a game is often described as the notion of \emph{price of anarchy} (PoA) and \emph{price of stability} (PoS).  They are defined as
\begin{definition}
	Given $2 \times 2$ game with payoff matrices $A$ and $B$, and a set of equilibrium system states $S \subseteq [0,1]^2$, the price of anarchy (PoA) and the price of stability (PoS) are defined as
	\begin{align*}
		&PoA(S) = \frac{\max_{(x,y)\in[0,1]^2} SW(x,y)}{\min_{(x,y)\in S} SW(x,y)},
		&PoS(S) = \frac{\max_{(x,y)\in[0,1]^2} SW(x,y)}{\max_{(x,y)\in S} SW(x,y)}
	\end{align*}
\end{definition}

\section{Our Model}
\subsection{Q-learning Dynamics}
In this paper, we are particularly interested in the scenario when both players' strategies are evolving under \emph{Q-learning dynamics}:
\begin{align} 
&\dot x_i = x_i \bigg[ (\vv{A} \vv{y})_i - \vv{x}^T \vv{A} \vv{y} + T_x \sum_{j} x_j \ln(x_j/x_i) \bigg], 
&\dot y_i = y_i \bigg[ (\vv{B} \vv{x})_i - \vv{y}^T \vv{B} \vv{x} + T_y \sum_{j} y_j \ln(y_j/y_i) \bigg] \label{eq:conti_dynamics}
\end{align}

The Q-learning dynamics has been studied because of its connection with multi-agent learning problems.  For example, it has been shown in \cite{Sato:2003aa,tuyls2003selection} that the Q-learning dynamics captures the system evolution of a repeated game, where each player learns his strategy through Q-learning and Boltzmann selection rules.  More details are provided in Appendix~\ref{sec:from_q_to_q}.

An important observation on the dynamics \eqref{eq:conti_dynamics} is that it demonstrates the exploration/ exploitation tradeoff \cite{tuyls2003selection}.
We can find that the right hand side of equation \eqref{eq:conti_dynamics} is composed of two parts.  The first part
$
x_i [ (\vv{A} \vv{y})_i - \vv{x}^T \vv{A} \vv{y} ]
$
is exactly the vector field of replicator dynamic \cite{sandholm2009evolutionary}.  Basically, the replicator dynamics drives the system to the state of higher utility for both players.  As a result, we can consider this as a selection process in terms of population evolutionary, or an exploitation process from the perspective of a learning agent. Then, for the second part
$
x_i [T_x  \sum_{j} x_j \ln(x_j/x_i)],
$
we show in the appendix that if the time derivative of $\vv{x}$ contains this part alone, this results in the increase of the system entropy.

The system entropy is a function that captures the randomness of the system.  From the population evolutionary perspective, the system entropy corresponds to the variety of the population.  As a result, this term can be considered as the mutation process.  The level of the mutation is controlled by the temperature parameters $T_x$ and $T_y$.  Besides, in terms of the reinforcement learning, this term can be considered as an exploration process, as it provides the opportunity for the agent to gain information about the action that does not look the best so far.

\subsection{Convergence of the Q-learning dynamics}
By observing the Q-learning dynamics \eqref{eq:conti_dynamics}, we can find that the interior rest points for the dynamics are exactly the QREs of the $2 \times 2$ game.
It is claimed in \cite{Kianercy:2012aa} without proof that the Q-learning dynamics for a $2 \times 2$ game converges to interior rest points of probability simplexes for any positive temperature $T_x>0$ and $T_y>0$. We  provide a formal proof in Appendix \ref{appendix:convergence}.
The idea is that for positive temperature, the system is dissipative and by leveraging the planar nature of the system it can be argued that it converges to fixed points.

\subsection{Rescaling the Payoff Matrix}\label{sec:rescaling}
At the end of this section, we discuss the transformation of the payoff matrices that preserves the dynamics in \eqref{eq:conti_dynamics}.  This idea is proposed in \cite{hofbauer2005learning} and  \cite{hofbauer1998evolutionary}, where the \emph{rescaling} of a matrix is defined as follows
\begin{definition}[\cite{hofbauer1998evolutionary}]
	$\vv{A}'$ and $\vv{B}'$ is said to be a rescaling of $\vv{A}$ and $\vv{B}$ if there exist constants $c_j,d_i$, and $\alpha>0$, $\beta>0$ such that $a_{ij}'=\alpha a_{ij}+c_j$ and $b_{ji}' = \beta b_{ji} + d_i$.
\end{definition}
It is clear that 
rescaling the game payoff matrices is equivalent to updating the temperature parameters of the two agents
in \eqref{eq:conti_dynamics}.
So, wlog it suffices to study the dynamics under the assumption that the $2 \times 2$ payoff matrices $\vv{A}$ and $\vv{B}$ are in the following \emph{diagonal form}. 
\begin{definition}
	Given $2 \times 2$ matrices $\vv{A}$ and $\vv{B}$, their diagonal form is defined as
	$$
	\vv{A}_D = \left( \begin{array}{cc}
	a_{11}-a_{21} & 0 \\
	0 & a_{22}-a_{12} \end{array}
	\right)  , \quad
	\vv{B}_D = \left( \begin{array}{cc}
	b_{11}-b_{21} & 0 \\
	0 & b_{22}-b_{12} \end{array}
	\right) 
	$$
\end{definition}

Note that although rescaling the payoff matrices to their diagonal form preserves the equilibria, it does not preserves the social optimality, i.e. the socially optimal strategy profile in the transformed game is not necessary the socially optimal strategy profile in the original game.

%
%

\section{Hysteresis Effect and Bifurcation Analysis} \label{sec:bif}
\subsection{Hysteresis effect in Q-learning dynamics: An example}
We begin our discussion with an example:
\begin{example}[Hysteresis effect] \label{ex:ex1}
	Consider a $2 \times 2$ game with reward matrices
	\begin{equation}
	A = \left( \begin{array}{cc}
	10 & 0 \\
	0 & 5 \end{array}
	\right), \quad
	B = \left( \begin{array}{cc}
	2 & 0 \\
	0 & 4 \end{array}
	\right)
	\label{eq:ex1}
	\end{equation}
	
	\begin{figure}[t]
		\centering
		\begin{minipage}[c]{0.4\linewidth}
			\centering
			\includegraphics[width=0.8\linewidth]{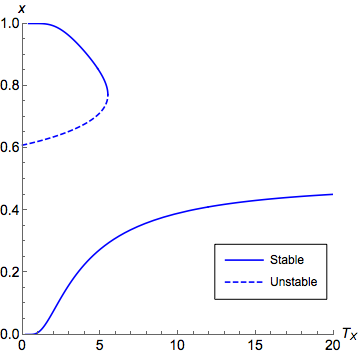}
			\caption{Example~\ref{ex:ex1}, $T_y=0.5$.}
			\label{fig:as_c1_low_TY}
		\end{minipage}
		\hspace{0.05\linewidth}
		\begin{minipage}[c]{0.4\linewidth}
			\centering
			\includegraphics[width=0.8\linewidth]{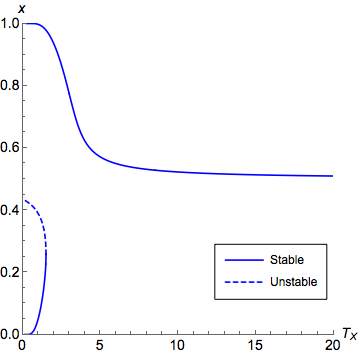}
			\caption{Example~\ref{ex:ex1}, $T_y=2$.}
			\label{fig:as_c1_high_TY}
		\end{minipage}
	\end{figure}
	
	There are two PNEs in this game, $(x,y)=(0,0)$ and $(1,1)$.  By fixing some $T_y$, we can plot different QREs with respect to $T_x$ as in Figure~\ref{fig:as_c1_low_TY} and Figure~\ref{fig:as_c1_high_TY}.  For simplicity, we only show the value of $x$ in the figure, since according to \eqref{eq:assym_xy}, given $x$ and $T_y$, the value of $y$ can be uniquely determined.  Assuming the system follows the Q-learning dynamics, as we slowly vary $T_x$, $x$ tends to stay on the line segment that is the closest to where it was originally corresponding to a stable but inefficient fixed point.  We consider the following process:
	\begin{enumerate}
		\item The initial state is $(0.05,0.14)$, where $T_x \approx 1$ and $T_y \approx 2$. We plot $x$ versus $T_x$ by fixing $T_y=2$ in Figure~\ref{fig:as_c1_high_TY}.
		\item Fix $T_y=2$, and increase $T_x$ to where there is only one QRE correspondence.
		\item Fix $T_y=2$, and decrease $T_x$ back to $1$. Now $x \approx 0.997$.
	\end{enumerate} 
\end{example}

In the above example, we can find that although at the end the temperature parameters are set back to their initial value, the system state ends up to be a totally different equilibrium.  This behavior is known as the \emph{hysteresis effect}.  In this section, we would like to answer the question \emph{when this is going to happen}.  Further, in the next section, we will answer \emph{how can we take advantage of this phenomenon}.

\subsection{Characterizing QREs}
We consider the bifurcation diagrams for QREs in $2\times 2$ games.  Without loss of generality, we consider a properly rescaled $2 \times 2$ game with payoff matrices in the diagonal form:
\begin{equation*} 
\vv{A}_D = \left( \begin{array}{cc}
a_X & 0 \\
0 & b_X \end{array}
\right), \quad
\vv{B}_D = \left( \begin{array}{cc}
a_Y & 0 \\
0 & b_Y \end{array}
\right)
\end{equation*}
Also, we can assume the action indices are ordered properly and rescaled properly so that $a_X>0$ and $|a_X|\ge|b_X|$. For simplicity, we assume $a_X=b_X$ and $b_X=b_Y$ do not hold at the same time.  At QRE, we have
\begin{align}
&x = \frac{e^{\frac{1}{T_x} ya_X}}{e^{\frac{1}{T_x} ya_X} + e^{\frac{1}{T_x} (1-y)b_X}},
&y = \frac{e^{\frac{1}{T_y} xa_Y}}{e^{\frac{1}{T_y} xa_Y} + e^{\frac{1}{T_y} (1-x)b_Y}}
\label{eq:assym_xy}
\end{align}

Given $T_x$ and $T_y$, there could be multiple solutions to \eqref{eq:assym_xy}.  
However, we find that if we know the equilibrium states, then we can recover the temperature parameters.  We solve for $T_x$ and $T_y$ in \eqref{eq:assym_xy}, and then we get 
\begin{align}
&T_X^I(x,y) = \frac{-(a_X+b_X) y + b_X}{\ln(\frac{1}{x}-1)},
&T_Y^I(x,y) = \frac{-(a_Y+b_Y) x + b_Y}{\ln(\frac{1}{y}-1)} \label{eq:Txy}
\end{align}
We call this the \emph{first form of representation}, where $T_x$ and $T_y$ are written as functions of $x$ and $y$.  Here the capital subscripts for $T_X$ and $T_Y$ indicate that they are considered as functions.  A direct observation to \eqref{eq:Txy} is that both of them are continuous function over $(0,1)\times(0,1)$ except for $x=1/2$ and $y=1/2$.

An alternative way to describe the QRE is to write $T_x$ and $y$ as a function of $x$ and parameterize with respect to $T_y$ in the following \emph{second form of representation}.  This will be the form that we use to prove many useful characteristics of QREs.
\begin{align}
T_X^{II}(x,T_y) &= \frac{-(a_X+b_X) y^{II}(x,T_y) + b_X}{\ln(\frac{1}{x}-1)} \label{eq:Tx_}\\
y^{II}(x,T_y) &= \bigg(1+e^{\frac{1}{T_y}(-(a_Y+b_Y) x + b_Y)}\bigg)^{-1} \label{eq:y_}
\end{align}
In this way, if we are given $T_y$, we are able to analyze how $T_x$ changes with $x$.  This helps us understand how to answer the question of what are the QREs given $T_x$ and $T_x$ in the system.

We also want to analyze the stability of the QREs.  From dynamical system theory (e.g. \cite{perko}), a fixed point of a dynamical system is said to be asymptotically stable if all of the eigenvalues of its Jacobian matrix has negative real part; if it has at least one eigenvalue with positive real part, then it is unstable.
It turns out that under the second form representation, we are able to determine whether a segment in the diagram is stable or not.
\begin{lemma} \label{lem:qre_stab}
	Given $T_y$, the system state $\left(x,y^{II}(x,T_y)\right)$ is a stable equilibrium if and only if
	\begin{enumerate}
		\item $\frac{\partial T_X^{II}}{\partial x}(x,T_Y)>0$ if $x \in (0, 1/2)$.
		\item $\frac{\partial T_X^{II}}{\partial x}(x,T_Y)<0$ if $x \in ( 1/2,1)$.
	\end{enumerate}
\end{lemma}
\begin{proof}
	The given condition is equivalent to the case that both eigenvalues of the Jacobian matrix of the dynamics \eqref{eq:conti_dynamics} are negative.
\end{proof}

Finally, we define the \emph{principal branch}.  In Example~\ref{ex:ex1}, we call the branch on $x \in (0.5,1)$ the \emph{principal branch} given $T_y=2$, since for any $T_x>0$, there is some $x \in (0.5,1)$ such that $T_X^{II}(x,T_y)=T_x$.  Analogously, we can define it formally as in the following definition with the help of the second form representation.
\begin{definition}
	Given $T_y$, the region $(a,b)\subset(0,1)$ \emph{contains the principal branch} of QRE correspondence if it satisfies the following conditions:
	\begin{enumerate}
		\item $T_X^{II}(x,T_y)$ is continuous and differentiable for $x \in (a,b)$.
		\item $T_X^{II}(x,T_y) > 0$ for $x \in (a,b)$.
		\item For any $T_x>0$, there exists $x \in (a,b)$ such that $T_X^{II}(x,T_y)=T_x$.
	\end{enumerate}
	Further, for a region $(a,b)$ that contains the principal branch, $x \in (a,b)$ is \emph{on the principal branch} if it satisfies the following conditions:
	\begin{enumerate}
		\item The equilibrium state $(x, y^{II}(x, T_y))$ is stable.
		\item There is no $x' \in (a,b), x'<x$ such that $T_X^{II}(x',T_y) =T_X^{II}(x,T_y) $.
	\end{enumerate}
\end{definition}

\subsection{Coordination Games}\label{sec:coord_topo}
We begin our analysis with the class of coordination games, where we have all $a_X$, $b_X$, $a_Y$, and $b_Y$ positive.    Also, without loss of generality, we assume $a_X \ge b_X$.  In this case, there are no dominant strategy for both players, and there are two PNEs.  

Let us revisit Example~\ref{ex:ex1}, we can make the following observations from Figure~\ref{fig:as_c1_low_TY} and Figure~\ref{fig:as_c1_high_TY}:
\begin{enumerate}
	\item Given $T_y$, there are three branches.  One is the principal branch, while the other two appear in pairs and occur only when $T_x$ is less than some value.
	\item For small $T_y$, the principal branch goes toward $x=0$; while for large $T_y$, the principal branch goes toward $x=1$.
\end{enumerate} 

Now, we are going to show that these observations are generally true in coordination games.  The proofs in this section are deferred to Appendix~\ref{appendix:bifurcation}, where we will give a detailed discussion on the proving techniques.  

The first idea we are going to introduce is the \emph{inverting temperature}, which is the threshold of $T_y$ in Observation~(2).  We  define it as 
$$T_I=\max\left\{0,\frac{b_Y-a_Y}{2\ln(a_X/b_X)}\right\}$$
We note that $T_I$ is positive only if $b_Y>a_Y$, which is the case that two players have different preferences.  When $T_y<T_I$, as the first player increases his rationality from fully irrational, i.e. $T_x$ decreases from infinity, he is likely to be influenced by the second player's preference.  If $T_y$ is greater than $T_I$, then the first player prefers to follow his own preference, making the principal branch goes toward $x=1$.  We formalize this idea in the following theorem:
\begin{theorem}[Direction of the principal branch]\label{thm:coord_topo_pb}
	Given a $2 \times 2$ coordination game, and given $T_y$, the following statements are true:
	\begin{enumerate}
		\item If $T_y>T_I$, then $(0.5,1)$ contains the principal branch.
		\item If $T_y<T_I$, then $(0,0.5)$ contains the principal branch.
	\end{enumerate}
\end{theorem}

The second idea is the \emph{critical temperature}, denoted as $T_C(T_y)$, which is a function of $T_y$.  The critical temperature is defined as the infimum of $T_x$ such that for any $T_x>T_C(T_y)$, there is a unique QRE correspondence under $(T_x, T_y)$. Generally, there is no close form for the critical temperature.  However, we can still compute it efficiently, as we show it in Theorem~\ref{thm:coord_topo_ct}.  Besides, another interesting value of $T_y$ we should be noticed is $T_B=\frac{b_Y}{\ln(a_X/b_X)}$, which is the maximum value of $T_y$ that QREs not on the principal branch are presenting.  Intuitively, as $T_y$ goes beyond $T_B$, the first player ignores the decision of the second player and turn his face to what he think is better.  We formalize the idea of $T_C$ and $T_B$ in the following theorem:

\begin{theorem}[Properties about the second QRE]\label{thm:coord_topo_ct}
	Given a $2 \times 2$ coordination game, and given $T_y$, the following statements are true:
	\begin{enumerate}
		\item For almost every $T_x>0$, all QREs not lying on the principal branch appear in pairs.
		\item If $T_y>T_B$, then there is no QRE correspondence in $x \in (0,0.5)$.
		\item If $T_y>T_I$, then there is no QRE correspondence for $T_x>T_C(T_y)$ in $x \in(0,0.5)$.
		\item If $T_y<T_I$, then there is no QRE correspondence for $T_x>T_C(T_y)$ in $x \in(0.5,1)$.
		\item $T_C(T_y)$ is given as $T_X^{II}(x_L,T_y)$, where $x_L$ is the solution to the equality
		$$
		y^{II}(x,T_y)+x(1-x)\ln\left(\frac{1}{x}-1\right)\frac{\partial y^{II}}{\partial x}(x,T_y)=\frac{b_X}{a_X+b_X}
		$$
		\item $x_L$ can be found using binary search.
	\end{enumerate}
\end{theorem}

The next aspect of the QRE correspondence is their stability.  According to Lemma~\ref{lem:qre_stab}, the stability of the QREs can also be inspected with the advantage of the second form representation by analyzing $\frac{\partial T_X^{II}}{\partial x}$.  We state the results in the following theorem:

\begin{theorem}[Stability]\label{thm:coord_topo_stab}
	Given a $2 \times 2$ coordination game, and given $T_y$, the following statements are true:
	\begin{enumerate}
		\item If $a_Y \ge b_Y$, then the principal branch is continuous.
		\item If $T_y<T_I$, then the principal branch is continuous.
		\item If $T_y>T_I$ and $a_Y< b_Y$, then the principal branch may not be continuous.
		\item Fix $T_x$, for the pairs of QREs not lying on the principal branch, the one of less distance to $x=0.5$ is unstable, while the other one is stable.
	\end{enumerate}
	
\end{theorem}

\begin{figure}[t]
	\centering
	\begin{minipage}[c]{0.4\linewidth}
		\centering
		\includegraphics[width=0.8\linewidth]{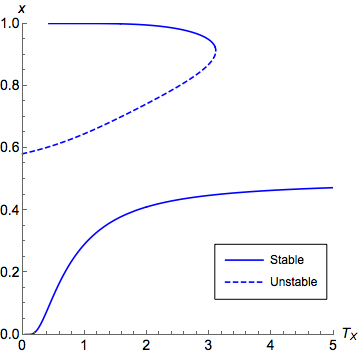}
		\caption{In a coordination game with $a_Y < b_Y$ and low $T_y$.}
		\label{fig:as_c1a_by_ex1}
	\end{minipage}
	\hspace{0.05\linewidth}
	\begin{minipage}[c]{0.4\linewidth}
		\centering
		\includegraphics[width=0.8\linewidth]{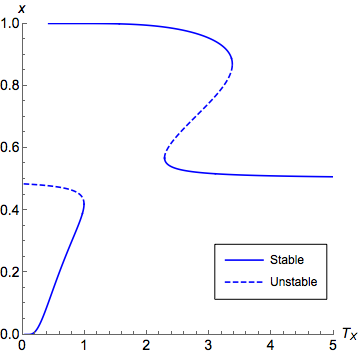}
		\caption{In a coordination game with $a_Y < b_Y$ and high $T_Y$.  There is a non-stable segment on the principal branch.}
		\label{fig:as_c1a_by_ex2}
	\end{minipage}
\end{figure}

Note that part 3 in Theorem~\ref{thm:coord_topo_stab} infers that there is potentially an unstable segment between segments of the principal branch.  This phenomenon is illustrated in Figure~\ref{fig:as_c1a_by_ex2}.  Though this case is weaker than other cases, this does not hinder us from designing a controlling mechanism as we are going to do in Section~\ref{sec:mec_1}.


\subsection{Non-coordination games}
Due to space constraint, the analysis for non-coordination games is deferred to Appendix~\ref{sec:topo_nc}.

\section{Mechanism Design}

\subsection{Hysteresis Mechanism: Select the Best Nash Equilibrium via QRE Dynamics}\label{sec:mec_2}
In this section, we consider the class of coordination games, and when the socially optimal state is one of the PNEs.  The main task for us in this case is to determine when and how we can get to the socially optimal PNE.  In Example~\ref{ex:ex1}, by sequentially changing $T_x$, we move the equilibrium state from around $(0,0)$ to around $(1,1)$, which is the social optimum state.  We formalize this idea as the \emph{hysteresis mechanism} and present it in Theorem~\ref{thm:mec2}. The hysteresis mechanism mainly takes advantage of the hysteresis effect we have discussed in Section~\ref{sec:bif}, that we use transient changes of system parameters to induce permanent
improvement to system performance via optimal equilibrium selection.

\begin{theorem}[Hysteresis Mechanism]\label{thm:mec2}
	Given a $2 \times 2$ game, if it satisfies the following property:
	\begin{enumerate}
		\item Its diagonal form satisfies $a_X,b_X,a_Y,b_Y>0$.
		\item Exactly one of its pure Nash equilibrium is the socially optimal state.
	\end{enumerate}
	Without loss of generality, we can assume $a_X \ge b_X$. Then, there is a mechanism to control the system to the social optimum by sequentially changing $T_x$ and $T_y$ if 1) $a_Y \ge b_Y$ and 2) the socially optimal state is $(0,0)$ do not hold at the same time.
\end{theorem}
\begin{proof}
	First, note that if $a_Y \ge b_Y$, by Theorem~\ref{thm:coord_topo_pb} the principal branch is always in the region $x > 0.5$.  As a result, once $T_y$ is increased beyond the critical temperature, the system state will no longer return to $x<0.5$ at any positive temperature.  Therefore, $(0,0)$ cannot be approached from any state in $x>0.5$ through the QRE dynamics.
	
	On the other hand, if $a_Y \ge b_Y$ and the socially optimal state is the PNE  $(1,1)$, then we can approach that state by first getting onto the principal branch.  The mechanism can be described as
	\begin{enumerate}
		\item[(C1)] \begin{enumerate}
			\item Raise $T_x$ to some value above the critical temperature $T_C(T_y)$.
			\item Reduce $T_x$ and $T_y$ to $0$.
		\end{enumerate}
	\end{enumerate}
	Though in this case, the initial choice of $T_y$ does not affect the result, if the social designer is taking the costs from assigning large $T_x$ and $T_y$ into account, he is going to trade off between $T_C$ and $T_y$ since typically smaller $T_y$ induces larger $T_C$.
	
	Next, consider $a_Y<b_Y$. 
	If we are aiming for state $(0,0)$, then we can do the following:
	\begin{enumerate}
		\item[(D1)] \begin{enumerate}
			\item Keep $T_y$ at some value below $T_I=\frac{b_Y-a_Y}{2\ln(a_X/b_X)}$. Now the principal branch is at $(0,0.5)$.
			\item Raise $T_x$ to some value above the critical temperature $T_C(T_y)$.
			\item Reduce $T_x$ to $0$.
			\item Reduce $T_y$ to $0$.
		\end{enumerate}
	\end{enumerate}
	On the other hand, if we are aiming for state $(1,1)$, then the following procedure suffices:
	\begin{enumerate}
		\item[(D2)] \begin{enumerate}
			\item Keep $T_y$ at some value above $T_I=\frac{b_Y-a_Y}{2\ln(a_X/b_X)}$.  Now the principal branch is at $(0.5,1)$.
			\item Raise $T_x$ to some value above the critical temperature $T_C(T_y)$.
			\item Reduce $T_x$ to $0$.
			\item Reduce $T_y$ to $0$.
		\end{enumerate}
	\end{enumerate}
	Note that in the last two steps only by reducing $T_y$ after $T_x$ keeps the state around $x=1$.  We recommend the reader to refer to Figure~\ref{fig:ex_rsw_p1} for case (D1), and Figure~\ref{fig:ex_rsw_p2} for case (D2) for more insights.
\end{proof}

\subsection{Efficiency of QREs: An example}
A question that arises with the solution concept of QRE is \emph{does QRE improves social welfare}?  Here we show that the answer is \emph{yes}.  We begin with an example to illustrate:
 \begin{example}\label{ex:ex2}
 	Consider a standard coordination game with the payoff matrices of the form
 	\begin{equation} \label{eq:ex_pos1}
 	A = \left( \begin{array}{cc}
 	\epsilon & 1 \\
 	0 & 1+\epsilon' \end{array}
 	\right), \quad
 	B = \left( \begin{array}{cc}
 	1+\epsilon & 0 \\
 	1 & \epsilon' \end{array}
 	\right)
 	\end{equation}
 	where $\epsilon > \epsilon' >0$ are some small numbers. Note that in this game, there are two PNEs $(x,y)=(1,1)$ and $(x,y)=(0,0)$, with social welfare $1+2\epsilon$ and $1+2\epsilon'$, respectively.  We can see that for small $\epsilon$ and $\epsilon'$, the socially optimal state is $(x,y)=(1,0)$, with social welfare $2$.  In this case, the state $(x,y)=(1,1)$ is the PNE with the best social welfare.  However, we are able to achieve the state with the better social welfare than any NE through QRE dynamics.  We illustrate the social welfare of the QREs with different temperatures of this example in Figure~\ref{fig:pos_improve}.  In this figure, we can see that at PNE, which is the point $T_x=T_y=0$, the social welfare is $1+2\epsilon$.  However, we are able to increase the social welfare by increasing $T_y$.  We will show in Section~\ref{sec:mec_1} a general algorithm to find the particular temperature, as well as a mechanism, which we refer to it as the \emph{optimal control mechanism}, that drives the system to the desired state.
 	\begin{figure}[t]
 		\centering
 		\begin{minipage}[c]{\linewidth}
 			\centering
 			\begin{minipage}[c]{0.4\linewidth}
 				\centering
 				\includegraphics[width=0.8\linewidth]{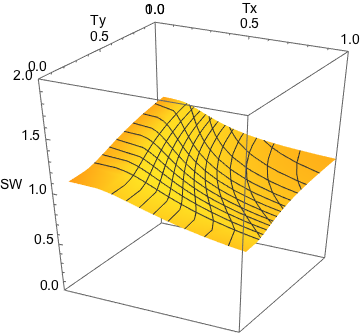}
 			\end{minipage}
 			\hspace{0.1\linewidth}
 			\begin{minipage}[c]{0.4\linewidth}
 				\centering
 				\includegraphics[width=0.8\linewidth]{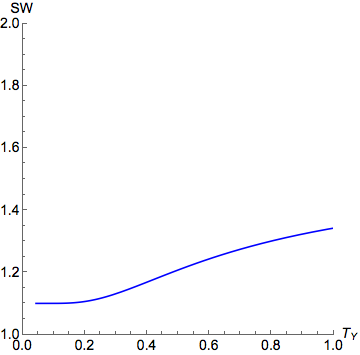}
 			\end{minipage}
 			\caption{The left figure is the social welfare on the principal branch for Example~\ref{ex:ex2}, and the right figure is the illustration when $T_X = 0$. We can see that by increasing $T_y$, we can obtain an equilibrium with the higher social welfare than the best Nash equilibrium (which is $T_x=T_y=0$)}
 			\label{fig:pos_improve}
 		\end{minipage}
 	\end{figure}
 	
 \end{example}


\subsection{Optimal Control Mechanism: Better Equilibrium with Irrationality}\label{sec:mec_1}
Here, we show a general approach to improve the PoS bound for coordination games from Nash equilibria by QREs and Q-learning dynamics.
We denote $QRE(T_x,T_y)$ as the set of QREs with respect to $T_x$ and $T_y$.  Further, denote $QRE$ as the set of the union of $QRE(T_x,T_y)$ over all positive $T_x$ and $T_y$.  Also, denote the set of pure Nash equilibria system states as $NE$.  Since the set $NE$ is the limit of $QRE(T_x,T_y)$ as $T_x$ and $T_y$ approach zero, we have the bounds:
$$
PoA(QRE) \ge PoA(NE), \quad PoS(QRE) \le PoS(NE)
$$
Then, we define \emph{QRE achievable states}:

\begin{definition}
	A state $(x,y) \in [0,1]^2$ is a QRE achievable state if for every $\epsilon>0$, there exist positive finite $T_x$ and $T_y$ and $(x',y')$ such that $|(x',y')-(x,y)|<\epsilon$ and $(x',y') \in QRE(T_x,T_y)$.
\end{definition}

\begin{figure}[t]
	\centering
	\begin{minipage}[c]{0.4\linewidth}
		\centering
		\includegraphics[width=0.8\linewidth]{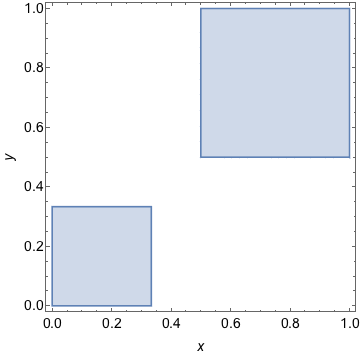}
		\caption{Set of QRE achievable states for Example~\ref{ex:ex2}.}
		\label{fig:ex1_region}
	\end{minipage}
	\hspace{0.05\linewidth}
	\begin{minipage}[c]{0.4\linewidth}
		\centering
		\includegraphics[width=0.8\linewidth]{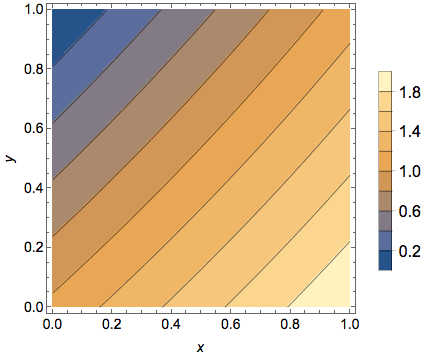}
		\caption{Social welfare for all states in Example~\ref{ex:ex2}.}
		\label{fig:ex1_sw_contour}
	\end{minipage}
\end{figure}

Note that with this definition, pure Nash equilibria are QRE achievable states.  However, the socially optimal states do not necessary to be QRE achievable.  For example, we illustrate in Figure~\ref{fig:ex1_region} the set of QRE achievable states for Example~\ref{ex:ex2}.  We can find that the socially optimal state, $(x,y)=(1,0)$, is not QRE achievable.  Nevertheless, it is easy to see from Figure~\ref{fig:ex1_region} and Figure~\ref{fig:ex1_sw_contour} that we can achieve a higher social welfare at $(x,y)=(1,0.5)$, which is a QRE achievable state.  Formally, we can describe the set of QRE achievable states as the positive support of $T_X^I$ and $T_Y^I$:
\begin{align*}
	S=&\left\{ \left\{x\in\left[\frac{1}{2},1\right], y\in\left[\frac{b_X}{a_X+b_X},1\right]\right\} \cup \left\{x\in\left[0,\frac{1}{2}\right], y\in\left[0,\frac{b_X}{a_X+b_X}\right]\right\} \right\} \\
	&\cap
	\left\{ \left\{x\in\left[\frac{b_Y}{a_Y+b_Y},1\right], y\in\left[\frac{1}{2},1\right]\right\} \cup \left\{x\in\left[0,\frac{b_Y}{a_Y+b_Y}\right], y\in\left[0,\frac{1}{2}\right]\right\} \right\}
\end{align*}
An example for the region of a game with $a_Y \ge b_Y$ is illustrated in Figure~\ref{fig:ex1_region}.  For the case $a_Y < b_Y$, we demonstrate it in Figure~\ref{fig:ex2_region}.

In the following theorem, we propose the \emph{optimal control mechanism} for a general process to achieve an equilibrium that is better than the PoS bound from Nash equilibria.


\begin{theorem}[Optimal Control Mechanism]\label{thm:mec1}
	Given a $2 \times 2$ game, if it satisfies the following property:
	\begin{enumerate}
		\item Its diagonal form satisfies $a_X,b_X,a_Y,b_Y>0$.
		\item None of its pure Nash equilibrium is the socially optimal state.
	\end{enumerate}
	Without loss of generality, we can assume $a_X \ge b_X$. Then,
	\begin{enumerate}
		\item there is a stable QRE achievable state whose social welfare is better than any Nash equilibrium.
		\item there is a mechanism to control the system to this state from the best Nash equilibrium by sequentially changing $T_x$ and $T_y$.
	\end{enumerate} 
\end{theorem}
\begin{proof}
	
	\begin{figure}[t]
		\centering
		\begin{minipage}[c]{0.4\linewidth}
			\centering
			\includegraphics[width=0.8\linewidth]{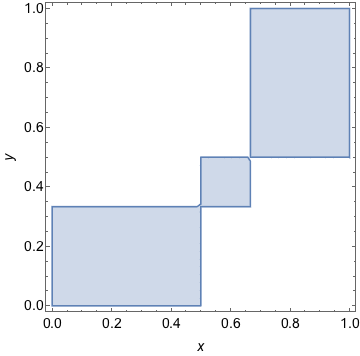}
			\caption{Set of QRE achievable states for a coordination game with $a_Y < b_Y$.}
			\label{fig:ex2_region}
		\end{minipage}
		\hspace{0.05\linewidth}
		\begin{minipage}[c]{0.4\linewidth}
			\centering
			\includegraphics[width=0.8\linewidth]{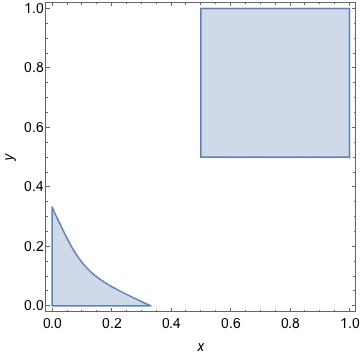}
			\caption{Stable QRE achievable states for a coordination game with $a_Y>b_Y$.}
			\label{fig:ex1_region_stab}
		\end{minipage}
	\end{figure}
	
	Note that given those properties, there are two PNEs $(0,0)$ and $(1,1)$.  Since we know neither of them is social optimum, the socially optimal state must lies on either $(0,1)$ or $(1,0)$.

	First, consider $a_Y \ge b_Y$.  In this case, we know from Theorem~\ref{thm:coord_topo_stab} that all $x \in (0.5,1)$ states belong to a principal branch for some $T_y>0$ and are stable.  While for $x < 0.5$, not all of them are stable.  We illustrate the region of stable QRE achievable states in Figure~\ref{fig:ex1_region_stab}. By Theorem~\ref{thm:coord_topo_ct} and Theorem~\ref{thm:coord_topo_stab}, we can infer that the states near the border $x=0$ are stable.  As a result, we can claim that the following states are what we are aiming for:
	\begin{enumerate}
		\item[(A1)] If $(1,1)$ is the best NE and $(0,1)$ is the SO state, then we select $(0.5,1)$.
		\item[(A2)] If $(1,1)$ is the best NE and $(1,0)$ is the SO state, then we select $(1,0.5)$.
		\item[(A3)] If $(0,0)$ is the best NE and $(0,1)$ is the SO state, then we select $\left(0,\frac{b_X}{a_X+b_X}\right)$.
		\item[(A4)] If $(0,0)$ is the best NE and $(1,0)$ is the SO state, then we select $\left(\frac{b_Y}{a_Y+b_Y},0\right)$.
	\end{enumerate}
	
	It is clear that these choices of states makes improvements on the social welfare.  It is known that for the class of games we are considering, the price of stability is no greater than $2$.  In fact, in case A1 and A2, we reduce this factor to $4/3$.  Also in case A3 and A4, we reduce this factor to $\left(\frac{1}{2}+\frac{b_X/2}{a_X+b_X}\right)^{-1}$.
	
	The next step is to show the mechanism to drive the system to the desired state.  Due to symmetry, we only discuss case A1 and A3, where case A2 and case A4 can be done analogously.  For case A1,  the state corresponds to the temperature $T_x \rightarrow \infty$ and $T_y \rightarrow 0$.  For any small $\delta>0$, we can always find the state $(0.5+\delta,1-\delta)$ on the principal branch of some $T_y$.  This means that we can achieve this state from any initial state, not only from the NEs.  With the help of the first form representation of the QREs in \eqref{eq:Txy}, given any QRE achievable system state $(x,y)$, we are able to recover them to corresponding temperatures through $T_X^I$ and $T_Y^I$.  The mechanism can be described as follows: 
	\begin{enumerate}
		\item[(A1)] \begin{enumerate}
			\item From any initial state, raise $T_x$ to $T_X^I(0.5+\delta,1-\delta)$.
			\item Decrease $T_y$ to $T_Y^I(0.5+\delta,1-\delta)$
		\end{enumerate}
	\end{enumerate}
	
	For case A3, the state we selected is not on the principal branch.  This means that we cannot increase the temperatures too much; otherwise the system state will move to the principal branch and will never return.  We assume initially the system state is at $(\delta,\delta)$ for some small $\delta>0$, which is some state close to the best NE.  Also, we can assume the initial temperatures are $T_x=T_X^I(\delta,\delta)$ and $T_y=T_Y^I(\delta,\delta)$.  Our goal is to arrive at the state $\left(\delta_1, \frac{b_X}{a_X+b_X}-\delta_2\right)$ for some small $\delta_1>0$ and $\delta_2>0$ such that $\left(\delta_1, \frac{b_X}{a_X+b_X}-\delta_2\right)$ is stable.  We present the mechanism in the following:
	\begin{enumerate}
		\item[(A3)] \begin{enumerate}
			\item From initial state $(\delta,\delta)$, move $T_x$ to $T_X^I\left(\delta_1, \frac{b_X}{a_X+b_X}-\delta_2\right)$.
			\item Increase $T_y$ to $T_Y^I\left(\delta_1, \frac{b_X}{a_X+b_X}-\delta_2\right)$
		\end{enumerate}
	\end{enumerate}
	
	Here note that Step (b) should not be proceeded before Step~(a) because as we increase $T_y$ first, then we are taking the risks of getting off to the principal branch.
	
	Next, consider the case that $a_Y<b_Y$.  Similarly to the previous case, we know from Theorem~\ref{thm:coord_topo_ct} and Theorem~\ref{thm:coord_topo_stab} that states near the borders $x=0,0.5,1$ and $y=0,0.5,1$ are basically stable states.  Hence, we can claim the following results:
	\begin{enumerate}
		\item[(B1)] If $(1,1)$ is the best NE and $(0,1)$ is the SO state, then we select $\left(\frac{b_Y}{a_Y+b_Y},1\right)$.
		\item[(B2)] If $(1,1)$ is the best NE and $(1,0)$ is the SO state, then we select $(1,0.5)$.
		\item[(B3)] If $(0,0)$ is the best NE and $(0,1)$ is the SO state, then we select $\left(0,\frac{b_X}{a_X+b_X}\right)$.
		\item[(B4)] If $(0,0)$ is the best NE and $(1,0)$ is the SO state, then we select $\left(0.5,0\right)$.
	\end{enumerate}
	
	It is clear that these choices of states create improvement on the social welfare.  An interesting result for this case is that basically these desired states can be reached from any initial state.  Due to symmetry, we demonstrate the mechanisms for case (B3) and (B4), and the remaining ones can be done analogously.
	
	For case (B3), we are aiming for the state $\left(\delta_1,\frac{b_X}{a_X+b_X}-\delta_2\right)$ for some small $\delta_1>0$ and $\delta_2>0$. We propose the following mechanism:
	\begin{figure}[t]
		\centering
		\begin{minipage}[c]{0.4\linewidth}
			\centering
			\includegraphics[width=0.8\linewidth]{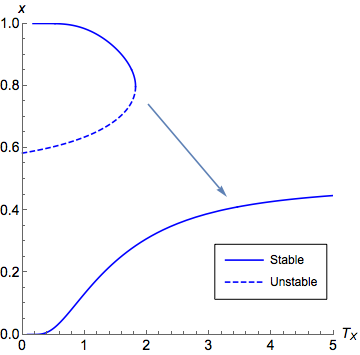}
			\caption{Phase 1 for case (B3), where we keep low $T_Y$ but increase $T_X$.}
			\label{fig:ex_rsw_p1}
		\end{minipage}
		\hspace{0.05\linewidth}
		\begin{minipage}[c]{0.4\linewidth}
			\centering
			\includegraphics[width=0.8\linewidth]{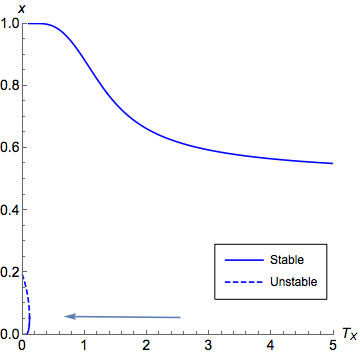}
			\caption{Phase 2 for case (B3), where we increase $T_Y$.}
			\label{fig:ex_rsw_p2}
		\end{minipage}
	\end{figure}
	
	\begin{enumerate}
		\item[(B3)] \begin{enumerate}
			\item[] \textbf{Phase 1:} Getting to the principal branch.
			\item From any initial state, fix $T_y$ at some value less than $T_I= \frac{b_Y-a_Y}{2\ln(a_X/b_X)}$.
			\item Increase $T_x$ above the critical temperature $T_C(T_y)$.
			\item Decrease $T_x$ to $T_x^I\left(\delta_1,\frac{b_X}{a_X+b_X}-\delta_2\right)$.
			\item[] \textbf{Phase 2:} Staying at the current branch.
			\item Increase $T_y$ to $T_Y^I\left(\delta_1,\frac{b_X}{a_X+b_X}-\delta_2\right)$.
		\end{enumerate}
	\end{enumerate}
	This process is illustrated in Figure~\ref{fig:ex_rsw_p1} and Figure~\ref{fig:ex_rsw_p2}.  In phase 1, as we are keeping low $T_y$, meaning the second player is of more rationality.  As the first player getting more rational, he is more likely to be influenced by the second player's preference, and eventually getting to a Nash equilibrium.  In phase 2, we make the second player more irrational to increase the social welfare.  The level of irrationality we add in phase 2 should be capped to prevent the first player to deviate his decision.
	
	For case (B4), since our desired state is on the principal branch, the mechanism will be similar to case (A1).
	\begin{enumerate}
		\item[(B4)] \begin{enumerate}
			\item From any initial state, raise $T_x$ to $T_X^I(0.5+\delta,\delta)$.  
			\item Decrease $T_y$ to $T_Y^I(0.5+\delta,\delta)$.
		\end{enumerate}
	\end{enumerate}
\end{proof}

	As a remark, in case (A3) and (A4), if we do not start from $(\delta,\delta)$ but from some other states on the principal branch, we can instead aim for state $(0.5,1)$.  This state is not better than the best Nash equilibrium, but still makes improvements over the initial state.  The process can be modified as:
		\begin{enumerate}
			\item[(A3')] \begin{enumerate}
				\item From any initial state, raise $T_x$ to $T_X^I(0.5+\delta,1-\delta)$ (above $T_C(T_y)$).
				\item Reduce $T_y$ to $T_Y^I(0.5+\delta,1-\delta)$.
			\end{enumerate}
		\end{enumerate}


\section{Applications}
\subsection{Taxation} \label{sec:taxation}
A direct application for the solution concept of QRE is to analyze the effect of taxation, which has been discussed in \cite{wolpert2012hysteresis}.  Unlike Nash equilibria, for QREs, if we multiply the payoff matrix by some factor $\alpha$, the equilibrium does change.  This is because by multiplying $\alpha$, effectively we are dividing the temperature parameters by $\alpha$.  This means that if we charge taxes to the players with some flat tax rate $\alpha-1$, the QREs will differ.  Formally, we define the base temperature $T_0$ as the temperature when no tax is applied for both players.  Then, we can define the \emph{tax rate} for each player as $\alpha_x = 1-T_0/T_x, \alpha_y=1-T_0/T_y$, respectively.

Now we demonstrate how the hysteresis mechanism can be applied via taxation with Example~\ref{ex:ex1}.  Assume the base temperature $T_0=1$, then with taxation, we can rewrite the process in Example~\ref{ex:ex1} in the following form:
\begin{enumerate}
	\item The initial state is $(0.05,0.14)$, where $\alpha_x \approx 0$ and $\alpha_y \approx 0.5$ (where $T_x \approx 1$ and $T_y \approx 2$).
	\item Fix $\alpha_y = 0.5$ (where $T_y=2$), and increase $\alpha_x$ to $0.8$, where $T_x=5$ and there is only one QRE correspondence.
	\item Fix $\alpha_y = 0.5$ (where $T_y=2$), and decrease $\alpha_x$ back to $0$ (where $T_x=1$). Now $x \approx 0.997$.
\end{enumerate} 

\subsection{Evolution of metabolic phenotypes in cancer} \label{sec:cancer}
Evolutionary Game Theory has been instrumental in studying evolutionary aspects of the somatic evolution that characterize's cancer progression. Tomlinson and Bodmer were the first to explore the role of cell-cell interactions in cancer. This pioneering work was followed by others that expanded on those initial ideas to study the role of key aspects of cancer evolution like the role of space \cite{kaznatcheev2015edge} treatment \cite{basanta2012investigating,kaznatcheev2016cancer} or metabolism \cite{basanta2008evolutionary,kianercy2014critical}. With regards to Kianercy's work, it shows how microenvironmental heterogeneity impacts somatic evolution, in this case by optimizing the genetic instability to better tune cell metabolism to the dynamic microenvironment.

Our techniques (the hysteresis mechanism and the optimal control mechanism) can be applied to the cancer game \cite{kianercy2014critical} with two types of tumor phenotypic strategies: hypoxic cells and oxygenated cells. These cells inhabit regions where oxygen could be abundant or lacking. In the former, oxygenated cells with regular metabolism thrive but in the latter, hypoxic cells whose metabolism is less reliant on the presence of oxygen (but more on the presence of glucose) have higher fitness. 

\vspace{-10pt}

\section{Connection to previous works}\label{sec.rel}



Recently, there has been a growing interplay between game theory, dynamical systems, and computer science.  Particular such examples include the integration of replicator dynamics and topological tools \cite{PiliourasAAMAS2014,papadimitriou2016nash,panageas2016average} in algorithmic game theory, and  Q-learning dynamics \cite{watkins1992q} in multi-agent systems \cite{tan1993multi}.  Q-learning dynamics has been studied extensively in game settings e.g. by Sato~\etal in \cite{Sato:2003aa} and Tuyls~\etal in \cite{tuyls2003selection}.  In \cite{coucheney2013entropy} Q-learning dynamics is considered as an extension of replicator dynamics driven by a combination of payoffs and  entropy. 
 Recent advances in our understanding of evolutionary dynamics in multi-agent learning can be found in the survey \cite{bloembergen2015evolutionary}.  

We are particularly interested in the connection between the Q-learning dynamics and the concept of QRE \cite{McKelvey:1995aa} in game theory.  In \cite{Cominetti:2010aa} Cominetti~\etal study this connection in traffic congestion games.  The hysteresis effect of Q-learning dynamics was first identified in 2012 by Wolpert~\etal \cite{wolpert2012hysteresis}.  Kianercy~\etal in \cite{Kianercy:2012aa} observed the same phenomenon, and provided discussions on the bifurcation diagrams in $2 \times 2$ games.  The hysteresis effect has been also been highlighted in recent follow-up work by \cite{kianercy2014critical} as a design principle for future cancer treatments.  It was also studied in \cite{romero2015effect} in the context of minimum-effort coordination games.
 However, our current understanding is still mostly qualitative and in this work we have pushed towards a more practically applicable quantitative, algorithmic analysis.

Analyzing the characteristics of various dynamical systems has also been attracting the attention of the computer science community in recent years.  For example, besides the Q-learning dynamics, the (simpler) replicator dynamics has been studied extensively due to its connections \cite{paperics11,papadimitriou2016nash,Soda14} to the multiplicative weight update (MWU) algorithm in \cite{Kleinberg09multiplicativeupdates}. 

Finally, a lot of attention has also been devoted to biological systems and their connections to game theory and computation.  In recent work by Mehta~\etal \cite{mehta_et_al:LIPIcs:2016:6407}, the connection with genetic diversity was discussed in terms of the complexity of predicting whether genetic diversity persists in the long run  under evolutionary pressures.  
This paper builds upon a rapid sequence of related results \cite{PNAS1:Livnat16122008,ITCS:DBLP:dblp_conf/innovations/ChastainLPV13,PNAS2:Chastain16062014,livnat2014satisfiability,Meir15,ITCS15MPP}. The key result is \cite{ITCS:DBLP:dblp_conf/innovations/ChastainLPV13,PNAS2:Chastain16062014} where effectively it was made clear that there exists a strong connection between studying replicator dynamics in games and standard models of evolution. Follow-up works show how to analyze dynamics that incorporate errors (i.e. mutations) \cite{MPPPV} and how these mutations can have a critical effect to ensuring survival in the presence of dynamically changing environments. Our paper makes  progress along these lines by examining how noisy dynamics can introduce such as bifurcations.

We were inspired by recent work by Kianercy~\etal establishing a  connection between cancer dynamics and  cancer treatment and studying Q-learning dynamics in games. This is analogous to the connections \cite{CACM,ITCS:DBLP:dblp_conf/innovations/ChastainLPV13,PNAS2:Chastain16062014} between MWU and evolution detailed above. It is our hope that by starting off a quantitative analysis of these systems we can kickstart similarly rapid developments in our understanding of the related questions.

\section{Conclusion}

In this paper, 
we perform a quantitative analysis of  bifurcation phenomena connected to Q-learning dynamics  in the class of $2 \times 2$ games. Based on this analysis, we introduce two novel mechanisms, the hysteresis mechanism and the optimal control mechanism.  
Hysteresis mechanisms use transient changes to the system parameters to induce permanent
improvements to its performance via optimal (Nash) equilibrium selection.  Optimal control
mechanisms induce convergence to states whose performance is better than even the best
Nash equilibrium, showing that by controlling the exploration/exploitation tradeoff we can
achieve strictly better states than those achievable by perfectly rational agents.

We believe that these new classes of mechanisms could lead to interesting and new questions  within game theory as well as a more thorough understanding of cancer biology.

%
%

\bibliographystyle{plainnat}
\bibliography{ref,sigproc2,sigproc4,p198mehta}

\newpage
\section{Supplementary materials}
\appendix
\section{From Q-learning to Q-learning Dynamics} \label{sec:from_q_to_q}
In this section, we provide a quick sketch on how we can get to the Q-learning dynamics from Q-learning agents.
We start with an introduction to the Q-learning rule.  Then, we discuss the multi-agent model when there are multiple learners in the system.  The goal for this section is to identify the dynamics of the system in which there are two learning agents playing a $2 \times 2$ game repeatedly over time.  

\subsection{Q-learning Introduction}
Q-learning \cite{watkins1992q,watkins1989learning} is a value-iteration method for solving the optimal strategies in Markov decision processes.  It can be used as a model where users learn about their optimal strategy when facing uncertainties.  Consider a system that consists of a finite number of states and there is one player who has a finite number of actions.  The player is going to decide his strategy over an infinite time horizon.  In Q-learning, at each time $t$, the player stores a value estimate $Q_{(s,a)}(t)$ for the payoff of each state-action pair $(s,a)$.  Then, he chooses his action $a_{t+1}$ that maximizes the $Q$-value $Q_{(s_t,\cdot)}(t)$ for time $t+1$, given the system state is $s_t$ at time $t$.  In the next time step, if the agent plays action $a_{t+1}$, he will receive a reward $r(t+1)$, and the value estimate is updated according to the rule:
$$
Q_{(s_t,a_{t+1})}(t+1) = (1-\alpha) Q_{(s_t,a_{t+1})}(t) + \alpha (r(t+1) + \gamma \max_{a'} Q_{(s_{t+1},a')}(t))
$$
where $\alpha$ is the step size, and $\gamma$ is the discount factor.

\subsection{Joint-learning Model}
Next, we consider the joint learning model as in \cite{Kianercy:2012aa}.  Suppose there are multiple players in the system that are learning concurrently.  Denote the set of players as $P$.  We assume the system state is a function of the action each player is playing, and the reward observed by each player is a function of the system state.  Their learning behaviors are modeled as simplified models based on the Q-learning algorithm described above.  More precisely, we consider the case that each player assumes the system is only of one state, which corresponds to the case that the player has very limited memory, and has discount factor $\gamma=0$. The reward observed by player $i\in P$ given he plays action $a$ at time $t$ is denoted as $r_a^i(t)$.  We can write the updating rule of the $Q$-value for agent $i$ as follows:
$$
Q_{a}^i(t+1) = Q_a^i(t) + \alpha [r_a^i(t) - Q_a^i(t)]
$$

For the selection process, we consider the mechanism that each player $i \in P$ selects his action according to the Boltzmann distribution with temperature $T_i$:
\begin{equation} \label{eq:boltzmann_sel}
x_a^i(t) = \frac{e^{Q_{a}^i(t)/T_i}}{\sum_{a'} e^{Q_{a'}^i(t)/T_i}}
\end{equation}
where $x_a^i(t)$ is the probability that agent $i$ chooses action $a$ at time $t$.  The intuition behind this mechanism is that we are modeling the irrationality of the users by the temperature parameter $T_i$.  For small $T_i$, the selection rule corresponds to the case of more rational agents.  We can see that for $T_i \rightarrow 0$, \eqref{eq:boltzmann_sel} corresponds to the best-response rule, that is, each agent selects the action with the highest $Q$-value with probability one.  On the other hand, for $T_i \rightarrow \infty$, we can see that \eqref{eq:boltzmann_sel} corresponds to the selection rule of selecting each action uniformly at random, which models the case of fully-irrational agents.

\subsection{Continuous-time dynamics}
This underlying Q-learning model has been studied in the previous decades.  
It is known that if we take the time interval to be infinitely small, this sequential joint learning process can be approximated as a continuous-time model (\cite{tuyls2003selection,Sato:2003aa}) that has some interesting characteristics.  To see this, consider the $2 \times 2$ game as we have described in Section~\ref{sec:def_2_by_2_game}.  The expected payoff for the first player at time $t$ given he chooses action $a$ can be written as $r_a^x(t)=[\vv{A} \vv{y}(t)]_a$, and similarly, the expected payoff for the second player at time $t$ given he chooses action $a$ is $r_a^y(t)=[\vv{B} \vv{x}(t)]_a$.  The continuous-time limit for the evolution of the $Q$-value for each player can be written as 
\begin{align*}
\dot Q^x_a(t) &= \alpha [r_a^x(t) - Q_a^x(t)]\\
\dot Q^y_a(t) &= \alpha [r_a^y(t) - Q_a^y(t)]
\end{align*}
Then, we take the time derivative of \eqref{eq:boltzmann_sel} for each player to get the evolution of the strategy profile:
\begin{align*}
\dot x_i &= \frac{1}{T_x} x_i \bigg( \dot Q_i^x - \sum_{k} x_k \dot Q_k^x\bigg)\\
\dot y_i &= \frac{1}{T_y} y_i \bigg( \dot Q_i^y - \sum_{k} y_k \dot Q_k^y\bigg)
\end{align*}
Putting these together, and rescaling the time horizon to $\alpha t/T_x$ and $\alpha t/T_y$ respectively, we obtain the continuous-time dynamics:
\begin{align} 
\dot x_i &= x_i \bigg[ (\vv{A} \vv{y})_i - \vv{x}^T \vv{A} \vv{y} + T_x \sum_{j} x_j \ln(x_j/x_i) \bigg] \label{eq:conti_dynamics1_r}\\
\dot y_i &= y_i \bigg[ (\vv{B} \vv{x})_i - \vv{y}^T \vv{B} \vv{x} + T_y \sum_{j} y_j \ln(y_j/y_i) \bigg] \label{eq:conti_dynamics2_r}
\end{align}

\subsection{The exploration term increases entropy}

Now, we show that the exploration term in the Q-learning dynamics results in the increase of the entropy:

\begin{lemma}\label{lem:incr_entropy}
	Suppose $A = \vv{0}$ and $B = \vv{0}$. The system entropy 
	$$ H(\vv{x},\vv{y}) = H(\vv{x}) + H(\vv{y}) = -\sum_{i} x_i \ln x_i - \sum_{i} y_i \ln y_i $$
	for the dynamics \eqref{eq:conti_dynamics} increases with time, i.e.
	$$ \dot H(\vv{x}, \vv{y}) > 0 $$
	if $\vv{x}$ and $\vv{y}$ are not uniformly distributed.
\end{lemma}
\begin{proof}[Proof of Lemma~\ref{lem:incr_entropy}]
	It is equivalent that we consider the single agent dynamics:
	$$
	\dot x_i = x_i T_x \bigg[ -\ln x_i + \sum_{j} x_j \ln x_j \bigg]
	$$
	Taking the derivative of the entropy $H(\vv{x})$, we have
	\begin{align*}
		\dot H(\vv{x}) &= \sum_{i} (-\ln x_i -1) \dot x_i 
		= -T_x \bigg[ -\sum_{i} x_i (\ln x_i)^2 + \bigg(\sum_{j} x_i \ln x_i\bigg)^2 \bigg]
	\end{align*}
	and since we have $\sum_{i} x_i =1$, by Jensen's inequality, we can find that
	$$
	\bigg(\sum_{j} x_i \ln x_i\bigg)^2 \le \sum_{i} x_i (\ln x_i)^2
	$$
	where equality holds if and only if $\vv{x}$ is a uniform distribution.  Consequently, if we have $x_i \in (0,1)$, and $\vv{x}$ is not a uniform distribution, $\dot H(\vv{x})$ is strictly positive, which means that the system entropy increases with time.
\end{proof}

\section{Convergence of dissipative learning dynamics in $2\times 2$ games}
\label{appendix:convergence}

\subsection*{Liouville's formula}

Liouville's formula can be applied to any system of autonomous differential equations with a continuously differentiable
vector field $V$ on an open domain of $\mathcal{S} \subset \R^k$.  
 The divergence of $V$ at $x \in \mathcal{S}$ is defined
as the trace of the corresponding Jacobian at $x$, \textit{i.e.},  $\text{div}[V(x)]\equiv\sum_{i=1}^k \frac{\partial V_i}{\partial x_i}(x)=tr(DV(x))$. 
Since divergence is a continuous function we can compute its integral over measurable sets $A\subset \mathcal{S}$ (with respect to Lebesgue measure $\mu$ on $\R^n$). Given 
any such set $A$, let $\phi_t(A)= \{\phi(x_0,t): x_0 \in A\}$ be the image of $A$ under map $\Phi$ at time $t$. $\phi_t(A)$ is measurable
and its measure is $\mu(\phi_t(A)))= \int_{\phi_t(A)}dx$. Liouville's formula states that the time derivative of the volume $\phi_t(A)$ exists and
is equal to the integral of the divergence over $\phi_t(A)$: $\frac{d}{dt} [A(t)] = \int_{\phi_t(A)} \text{div} [V(x)]dx.$ Equivalently:

\begin{theorem}[\cite{Sandholm10}, page 356] 
$\frac{d}{dt}\mu(\phi_t(A))= \int_{\phi_t(A)} tr(DV(x))d\mu(x)$
\end{theorem}

A vector field is called divergence free if its divergence is zero everywhere. Liouville's formula trivially implies that volume is preserved
in such flows.

This theorem extends in a straightforward manner to systems where the vector field $V:X\rightarrow TX$ is defined on an affine set $X\subset \R^n$ with tangent space $TX$. In this case, $\mu$ represents the Lebesgue measure on the (affine hull) of $X$. Note that the derivative of $V$ at a state $x \in X$ must be represented using the derivate matrix $DV(x) \in \R^{n\times n}$, which by definitions has rows in $TX$. If $\hat{V}:\R^n\rightarrow R^n$ is a $C^1$ extension of $V$ then $DV(x)=D\hat{V}(x)P_{TX}$, where $P_{TX}\in \R^{n\times n}$ is the orthogonal projection\footnote{To find the matrix of the orthogonal projection onto $TX$ (or any subspace $Y$ of $\R^n$) it suffices to find a basis ($\vec{v_1}, \vec{v_2}, \dots, \vec{v_m}$). Let $B$ be the matrix with columns $\vec{v_i}$, then $P = B(B^T B)^{-1}B^T$.} of $\R^n$ onto the subspace $TX$.

\subsection*{Poincar\'{e}-Bendixson theorem}

The Poincar\'{e}-Bendixson theorem is a powerful theorem that implies that two-dimensional systems cannot effectively exhibit chaos.
Effectively, the limit behavior is either going to be an equilibrium, a periodic orbit, or a closed loop, punctuated by one (or more) fixed points.
Formally, we have:

\begin{theorem}
[\cite{bendixson1901courbes,teschl2012ordinary}]
Given a differentiable real dynamical system defined on an open subset of the plane, then every non-empty compact $\omega$-limit set of an orbit, which contains only finitely many fixed points, is either
a fixed point, a periodic orbit, or a connected set composed of a finite number of fixed points together with homoclinic and heteroclinic orbits connecting these.
\end{theorem}

\subsection*{Bendixson-Dulac theorem}

By excluding the possibility of closed loops (i.e., periodic orbits, homoclinic cycles, heteronclinic cycles) we can effectively establish global convergence to equilibrium.
The following criterion, which was first established by Bendixson in 1901 and further refined by French mathematician Dulac in 1933, allows us to do that.
It is typically referred to as the Bendixson-Dulac negative criterion. It focus exactly on planar system where the measure of initial conditions always shrinks (or always increases) with time, i.e., dynamical systems with vector fields whose divergence is always negative (or always positive).

\begin{theorem}[\cite{muller2015methods}, page 210]
Let $D\subset \R^2$ be a simply connected region and $(f,g)$ in $C^1(D, \R)$ with $div(f,g)=\frac{\partial f}{\partial x}+\frac{\partial g}{\partial y}$ being not identically zero and without change of sign in $D$. Then the system 

$$\frac{ dx }{ dt } = f(x,y),$$

$$\frac{ dy }{ dt } = g(x,y)$$

\noindent
has no loops lying entirely in $D$.
\end{theorem}

\noindent
The function $ \varphi(x, y)$ is typically called the Dulac function.

\noindent
\textbf{Remark:} This criterion can also be generalized. Specifically, it holds for the system:

$$\frac{ dx }{ dt } =\rho(x,y) f(x,y),$$

$$\frac{ dy }{ dt } =\rho(x,y) g(x,y)$$

\noindent
if $\rho(x,y)>0$ is continuously differentiable. Effectively, we are allowed to rescale the vector field by a scalar function (as long as this function does not have any zeros), before we prove that the divergence is positive (or negative). That is, it suffices to find $\rho(x,y)>0$ continuously differentiable, such that $(\rho(x,y) f(x,y))_x+ (\rho(x,y) g(x,y))_y$ possesses a fixed sign.


By \cite{Kianercy:2012aa} we have that the after a change of variables, $u_k=\frac{\ln(x_{k+1})}{\ln x_1}$, $v_k=\frac{\ln(y_{k+1})}{\ln y_1}$ for $k=1, \dots, n-1$,
the replicator system transforms to the following system:

$$\dot{u}_k=\frac{\sum_j \hat{a}_kje^{v_j}}{1+\sum_j e^{v_j}}- T_xu_k, \dot{v}_k=\frac{\sum_j \hat{a}_kje^{u_j}}{1+\sum_j e^{u_j}}- T_xv_k, (\text{II})$$
where
$\hat{a}_{kj} = a_{k+1,j+1} - a_{1,j+1}$, $\hat{b}_{kj} = b_{k+1,j+1} - a_{1,j+1}$.

In the case of $2\times2$ games, we can apply both the Poincar\'{e}-Bendixson theorem as well as the Bendixson-Dulac theorem, since the resulting dynamical system is planar and $\frac{\partial\dot{u}_1}{\partial u_1}+\frac{\partial \dot{v}_1}{\partial v_1}=-(T_x+T_y)<0$. Hence, for any initial condition system (II) converges to equilibria. 
The flow of original replicator system in the $2\times2$ game is \textit{diffeomorhpic}\footnote{ A function $f$ between two topological spaces is called a \textit{diffeomorphism} if it has the following properties:
$f$ is a bijection, $f$ is continuously  differentiable, and $f$ has a continuously differentiable inverse. Two flows $\Phi^t:A\rightarrow A$ and $\Psi^t :B\rightarrow B$ are \textit{diffeomorhpic} if there exists a diffeomorphism $g : A \rightarrow B$ such that for each $x \in A$ and $t \in \R$
$g(\Phi^t (x)) = \Psi^t (g(x))$. If two flows are diffeomorphic then their vector fields are related by the derivative of the conjugacy. 
 That is, we get precisely the  same result that we would have obtained if we simply transformed the coordinates in their differential equations~\cite{Meiss2007}.}  to the flow of system (II), thus replicator dynamics with positive temperatures $T_x,T_y$ converges to equilibria for all initial conditions as well.  
 
\section{Bifurcation Analysis for Games with Only One Nash Equilibrium}\label{sec:topo_nc}
In this section, we present the results for the class of games with only one Nash equilibrium, where it can be either a pure one or a mixed one, where the mixed Nash equilibrium is defined as
\begin{definition}[mixed Nash equilibrium]
	A strategy profile $(x_{NE},y_{NE})$ is a mixed Nash equilibrium if
	\begin{align*}
		&x_{NE} \in \arg\max_{x \in [0,1]} \vv{x}^T\vv{A}\vv{y}_{NE},
		&y_{NE} \in \arg\max_{y \in [0,1]} \vv{y}^T\vv{B}\vv{x}_{NE}
	\end{align*}
\end{definition}
This corresponds to the case that at least one of $b_X$, $a_Y$, or $b_Y$ being negative.  Similarly, our analysis is based on the second form representation described in \eqref{eq:Tx_} and \eqref{eq:y_}, which demonstrates insights from the first player's perspective.

\subsection{No dominating strategy for the first player}
More specifically, this is the case when there is no dominating strategy for the first player, i.e. both $a_X$ and $b_X$ are positive.  From \eqref{eq:y_} we can presume that the characteristics of the bifurcation diagrams depends on the value of $a_Y+b_Y$ since it affects whether $y^{II}$ is increasing with $x$ or not.  Also, we can find some interesting phenomenon from the discussion below.

First, we consider the case when $a_Y+b_Y>0$.  This can be considered as a more general case as we have discussed in Section~\ref{sec:coord_topo}.   In fact, the statements we have made in Theorem~\ref{thm:coord_topo_pb}, Theorem~\ref{thm:coord_topo_ct}, and Theorem~\ref{thm:coord_topo_stab} applies to this case.  However, there are some subtle difference we should be noticed.  If $a_Y>b_Y$, where we can assume $b_Y<0$, then by the second part of Theorem~\ref{thm:coord_topo_ct}, there are no QRE in $x \in (0,0.5)$, since $T_B$ now is a negative number.  This means that we always only have the principal branch.  On the other hand, if $a_Y<b_Y$, where we can assume $a_Y<0$, then similar to the example in Figure~\ref{fig:as_c1a_by_ex1} and Figure~\ref{fig:as_c1a_by_ex2}, there could still be two branches.  However, we can presume that the second branch vanishes \emph{before} $T_y$ actually goes to zero, as the state $(1,1)$ is not a Nash equilibrium.
\begin{theorem}\label{thm:nc1}
	Given a $2\times 2$ game in which the diagonal form has $a_X, b_X>0$, $a_Y+b_Y>0$, and $a_Y<b_Y$, and given $T_y$, if $T_y<T_A$, where $T_A=\frac{-a_Y}{\ln(a_Y/b_Y)}$, then there are no QRE correspondence in $x \in (0.5,1)$.
\end{theorem}
The proof of the above theorem directly follows from Proposition~\ref{prop:c1a_g_right2} in the appendix.  An interesting observation here is that we can still make the first player get to his desired state by changing $T_y$ to some value that is greater than $T_A$.

\begin{figure}[t]
	\centering
	\begin{minipage}[c]{0.4\linewidth}
		\centering
		\includegraphics[width=\linewidth]{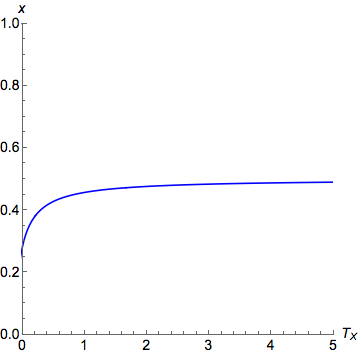}
		\caption{No dominating strategy for the first player, with $a_Y+b_Y<0$ and low $T_Y$.}
		\label{fig:as_c1b_low_TY}
	\end{minipage}
	\hspace{0.05\linewidth}
	\begin{minipage}[c]{0.4\linewidth}
		\centering
		\includegraphics[width=\linewidth]{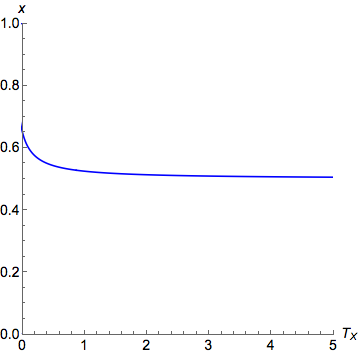}
		\caption{No dominating strategy for the first player, with $a_Y+b_Y<0$ and high $T_Y$.}
		\label{fig:as_c1b_high_TY}
	\end{minipage}
\end{figure}

Next, we consider $a_Y+b_Y\le 0$.  The bifurcation diagram is illustrated in Figure~\ref{fig:as_c1b_low_TY} and Figure~\ref{fig:as_c1b_high_TY}.  We can find that in this case the principal branch directly goes toward its unique Nash equilibrium.  We present the results formally in the following theorem, where the proof follows from Section~\ref{sec:case1b} in the appendix.
\begin{theorem}\label{thm:nc2}
	Given a $2\times 2$ game in which the diagonal form has $a_X, b_X>0$, $a_Y+b_Y \le 0$, QRE is unique given $T_x$ and $T_y$.
\end{theorem}

\subsection{Dominating strategy for the first player}
\begin{figure}[t]
	\centering
	\begin{minipage}[c]{0.4\linewidth}
		\centering
		\includegraphics[width=\linewidth]{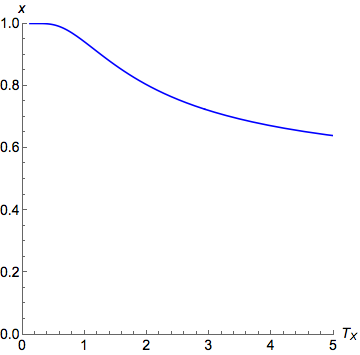}
		\caption{When there is a dominating strategy for the first player, with $a_Y+b_Y<0$.}
		\label{fig:as_c2_ex1}
	\end{minipage}
	\hspace{0.05\linewidth}
	\begin{minipage}[c]{0.4\linewidth}
		\centering
		\includegraphics[width=\linewidth]{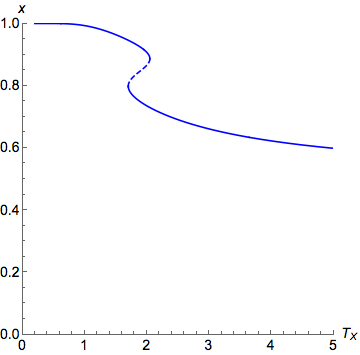}
		\caption{When there is a dominating strategy for the first player, with $a_Y+b_Y>0$ and $a_Y<b_Y$.}
		\label{fig:as_c2_ex2}
	\end{minipage}
\end{figure}

Finally, we consider the case when there is a dominating strategy for the first player, i.e. $b_X<0$.  According to Figure~\ref{fig:as_c2_ex1} and Figure~\ref{fig:as_c2_ex2}, the principal branch seems always goes towards $x=1$.  This means that the first player always prefers his dominating strategy.  We formalize this observation, as well as some important characteristics for this case in the theorem below, where the proof can be found in Section~\ref{sec:case2} in the appendix.
\begin{theorem}\label{thm:nc3}
	Given a $2\times 2$ game in which the diagonal form has $a_X>0$, $b_X<0$, $a_X+b_X>0$, and given $T_y$, the following statements are true:
	\begin{enumerate}
		\item The region $(0,0.5)$ contains the principal branch.
		\item There are no QRE correspondence for $x \in (0.5,1)$.
		\item If $a_Y+b_Y<0$ or $a_Y>b_Y$, then the principal branch is continuous.
		\item If $a_Y+b_Y>0$ and $b_Y>a_Y$, then the principal branch may not be continuous.
	\end{enumerate}
\end{theorem}

As we can see from Theorem~\ref{thm:nc3}, for the most cases, the principal branch is continuous.  One special case is when $a_Y+b_Y>0$ with $b_Y>a_Y$.  In fact, this can be seen as a duality, i.e. flipping the role of two players, of the case we have discussed in part 3 of Theorem~\ref{thm:nc1}, where for $T_y$ is within $T_A$ and $T_I$, there can be three QRE correspondences.
\section{Detailed Bifurcation Analysis for General $2 \times 2$ Game} \label{appendix:bifurcation}

In this section, we provide technical details for the results we stated in Section~\ref{sec:coord_topo} and Section~\ref{sec:topo_nc}.
Before we get into details, we state some results that will be useful throughout the analysis in the following lemma.  The proof of this lemma is straightforward and we omit it in this paper.

\begin{lemma}\label{lem:techs}
	The following statements are true.
	\begin{enumerate}
		\item The derivative of $T_X^{II}$ is given as
		\begin{equation} \label{eq:def_dTX}
		\frac{\partial T_X^{II}}{\partial x}(x,T_y) 
		= \frac{-(a_X+b_X)L(x,T_y)+b_X}{x(1-x)[\ln(1/x-1)]^2}
		\end{equation}
		where
		\begin{equation} \label{eq:def_L}
		L(x,T_y) = y^{II} + x(1-x)\ln\left(\frac{1}{x}-1\right)\frac{\partial y^{II}}{\partial x}
		\end{equation}
		\item The derivative of $y^{II}$ is given as
		$$
		\frac{\partial y^{II}}{\partial x} = y^{II}(1-y^{II})\frac{a_Y+b_Y}{T_y} 
		$$
		\item For $x \in (0,1/2)\cup(1/2,1)$, $\frac{\partial T_X^{II}}{\partial x}>0$ if and only if $L(x,T_y)<\frac{b_X}{a_X+b_X}$; on the other hand, $\frac{\partial T_X^{II}}{\partial x}<0$ if and only if $L(x,T_y)>\frac{b_X}{a_X+b_X}$.
	\end{enumerate}
\end{lemma}

\subsection{Case 1: $b_X \ge 0$}

First, we consider the case $b_X \ge 0$.  
As we are going to show in Proposition~\ref{prop:dir_pb}, the direction of the principal branch relies on $y^{II}(0.5, T_y)$, which is the strategy the second player is performing, assuming the first player is indifferent to his payoff.  The idea is that if $y^{II}(0.5, T_y)$ is large, then it means that the second player pays more attention to the action that the first player thinks better.  This is more likely to happen when the second player has less rationality, i.e. high temperature $T_y$.  On the other hand, if the second player pays more attention to the other action, the first player is forced to choose that as it gets more expected payoff.

We show that for $T_y>T_I$, the principal branch lies on $x\in\left(\frac{1}{2},1\right)$, otherwise the principal branch lies on $x\in\left(0, \frac{1}{2}\right)$.  This result follows from the following proposition:
\begin{proposition} \label{prop:dir_pb}
	For case~1, if $T_y>T_I$, then we have $y^{II}(1/2,T_y)>\frac{b_X}{a_X+b_X}$, and hence
	\begin{align*}
	\lim_{x\rightarrow \frac{1}{2}^+} T_X^{II}(x,T_y) = +\infty \quad \text{ and } \quad
	\lim_{x\rightarrow \frac{1}{2}^-} T_X^{II}(x,T_y) = -\infty
	\end{align*} On the other hand, if $T_y<T_I$, then we have $y^{II}(1/2,T_y)<\frac{b_X}{a_X+b_X}$, and hence
	\begin{align*}
	\lim_{x\rightarrow \frac{1}{2}^+} T_X^{II}(x,T_y) = -\infty \quad \text{ and } \quad \lim_{x\rightarrow \frac{1}{2}^-} T_X^{II}(x,T_y) = +\infty.
	\end{align*} 
\end{proposition}
\begin{proof}
	First, consider the case that $b_Y > a_Y$, then, we can see that for $T_y>T_I=\frac{b_Y-a_Y}{2\ln(a_X/b_X)}$:
	$$
	y^{II}\left(\frac{1}{2}, T_y\right) = \left( 1+e^{\frac{b_Y-a_Y}{2T_y}} \right)^{-1} > \left( 1+e^{\frac{b_Y-a_Y}{2T_I}} \right)^{-1} = \left(1+\frac{a_X}{b_X}\right)^{-1} = \frac{b_X}{a_X+b_X}
	$$
	Then, for the case that $a_Y>b_Y$, we can see that
	$$
	y^{II}\left(\frac{1}{2}, T_y\right) = \left( 1+e^{\frac{b_Y-a_Y}{2T_y}} \right)^{-1} > \left( 1+e^{0} \right)^{-1} = \frac{1}{2} \ge \frac{b_X}{a_X+b_X}
	$$
	For the case that $a_Y=b_Y$, since we assumed $a_X \not= b_X$, we have
	$$
	y^{II}\left(\frac{1}{2}, T_y\right) = \left( 1+e^{\frac{b_Y-a_Y}{2T_y}} \right)^{-1} = \left( 1+e^{0} \right)^{-1} = \frac{1}{2} > \frac{b_X}{a_X+b_X}
	$$
	As a result, the numerator of \eqref{eq:Tx_} at $x=\frac{1}{2}$ is negative for $T_y>T_I$, this proves the first two limit.  
	
	For the rest two limits, we only need to consider the case $b_Y>a_Y$, otherwise $T_I=0$, which is meaningless.  For $b_Y>a_Y$ and $T_y<T_I$, we can see that
	$$
	y^{II}\left(\frac{1}{2}, T_y\right) = \left( 1+e^{\frac{b_Y-a_Y}{2T_y}} \right)^{-1} < \left( 1+e^{\frac{b_Y-a_Y}{2T_I}} \right)^{-1} = \left(1+\frac{a_X}{b_X}\right)^{-1} = \frac{b_X}{a_X+b_X}
	$$
	This makes the numerator of \eqref{eq:Tx_} at $x=\frac{1}{2}$ positive and proves the last two limits.
\end{proof}

\subsubsection{Case 1a: $b_X\ge 0$, $a_Y+b_Y>0$}
In this section, we consider a relaxed version of the class of coordination game as in Section~\ref{sec:coord_topo}.  We prove theorems presented in Section~\ref{sec:coord_topo}, and showing that these results can in fact be extended to the case that $a_Y+b_Y>0$, instead of requiring $a_Y>0$ and $b_Y>0$.

First, we can find that as $a_Y+b_Y>0$, $y^{II}$ is an increasing function of $x$, meaning
$$
\frac{\partial y^{II}}{\partial x} = y^{II}(1-y^{II})\frac{a_Y+b_Y}{T_y} > 0
$$
This implies that both player tend to agree to each other.  Intuitively, if $a_Y \ge b_Y$, then both player agree with that the first action is the better one.  For this case, we can show that no matter what $T_y$ is, the principal branch lies on $x \in \left(\frac{1}{2},1\right)$.  In fact, this can be extended to the case whenever $T_y>T_I$, which is the first part of Theorem~\ref{thm:coord_topo_pb}.
\begin{proof}[Proof of Part~1 of Theorem~\ref{thm:coord_topo_pb}]
	We can find that for $T_y>T_I$, we have $y^{II}(1/2,T_Y)>\frac{b_X}{a_X+b_X}$ for any $T_y$ according to Proposition~\ref{prop:dir_pb}.  Since $y^{II}$ is monotonic increasing with $x$, we have $y^{II}>\frac{b_X}{a_X+b_X}$ for $x>1/2$.  This means that we have $T_X^{II}>0$ for any $x \in (1/2,1)$.  Also, it is easy to see that $\lim_{x\rightarrow 1^-} T_X^{II}=0$.  As a result, we can find that $(0.5,1)$ contains the principal branch.
\end{proof}


For Case~1a with $a_Y \ge b_Y$ we can observe that on the principal branch, the lower the $T_x$, the more $x$ is close to $1$.  We are able to show this monotonicity characteristics in Proposition~\ref{prop:c1a_g_monotone}, which can be used to justify the stability owing to Lemma~\ref{lem:qre_stab}.

\begin{proposition}\label{prop:c1a_g_monotone}
	In Case~1a, if $a_Y \ge b_Y$, then $\frac{\partial T_X^{II}}{\partial x}<0$ for $x \in \left(\frac{1}{2},1\right)$.
\end{proposition}
\begin{proof}
	It suffices to show that $L(x,T_y)>\frac{b_X}{a_X+b_X}$ for $x \in \left(\frac{1}{2},1\right)$.  Note that according to Prop~\ref{prop:dir_pb}, we have if $a_Y \ge b_Y$,
	\begin{equation}\label{eq:y_g_half}
	L(1/2,T_y) = y^{II}(1/2,T_y) \ge \frac{1}{2} 
	\ge \frac{b_X}{a_X+b_X}
	\end{equation}  
	Since $y^{II}(x,T_y)$ is monotonic increasing when $a_Y+b_Y>0$, $y^{II}(x,T_y)>\frac{1}{2}$ for $x \in \left(\frac{1}{2},1\right)$.  As a result, we have $1-2y^{II}<0$, and hence we can see that for $x \in \left(\frac{1}{2},1\right)$,
	$$
	\frac{\partial L}{\partial x} = \left[(1-2x)+x(1-x)(1-2y^{II})\frac{a_Y+b_Y}{T_y}\right]\ln\left(\frac{1}{x}-1\right)\frac{\partial y^{II}}{\partial x} > 0
	$$
	Consequently we have that for $x \in \left(\frac{1}{2},1\right)$,  $L(x,T_y)>\frac{b_X}{a_X+b_X}$, and hence $\frac{\partial T_X^{II}}{\partial x}<0$ according to Lemma~\ref{lem:techs}.
\end{proof}
\begin{proof}[Proof of Part 1 of Theorem~\ref{thm:coord_topo_stab}]
	According to Lemma~\ref{lem:qre_stab}, Proposition~\ref{prop:c1a_g_monotone} implies that all $x \in (0.5,1)$ is on the principal branch.  This directly leads us to part 1 of Theorem~\ref{thm:coord_topo_stab}.
\end{proof}

Next, if we look into the region $x \in (0,1/2)$, we can find that in this region, QREs appears only when $T_x$ and $T_y$ is low.  This observation can be formalized in the proposition below.  We can see that this proposition directly proves part~2 and 3 of Theorem~\ref{thm:coord_topo_ct}, as well as part~2 of Theorem~\ref{thm:coord_topo_stab}.

\begin{proposition}\label{prop:case1a_left}
	Consider Case~1a. Let $x_1 = \min\left\{\frac{1}{2},\frac{-T_y\ln\left(\frac{a_X}{b_X}\right)+b_Y}{a_Y+b_Y}\right\}$ and $T_{B} = \frac{b_Y}{\ln(a_X/b_X)}$. The following statements are true for $x \in (0,1/2)$:
	\begin{enumerate}
		\item If $T_y>T_B$, then $T_X^{II}<0$.
		\item If $T_y<T_B$, then $T_X^{II}>0$ if and only if $x \in (0, x_1)$.
		\item $\frac{\partial L}{\partial x} >0$ for $x \in (0,x_1)$.
		\item If $T_y<T_I$, then $\frac{\partial T_X^{II}}{\partial x}>0$.
		\item If $T_y>T_I$, then there is a nonnegative \emph{critical temperature} $T_C(T_y)$ such that $T_X^{II}(x, T_Y) \le T_{C}(T_y)$ for $x \in (0, 1/2)$.  If $T_Y<T_B$, then $T_C(T_y)$ is given as $T_X^{II}(x_L)$, where $x_L \in (0, x_1)$ is the unique solution to $L(x,T_y)=\frac{b_X}{a_X+b_X}$.
	\end{enumerate}
\end{proposition}
\begin{proof}
	For the first and second part, consider any $x \in (0, 1/2)$, and we can see that
	\begin{align*}
	T_X^{II} > 0 &\Leftrightarrow y^{II} < \frac{b_X}{a_X+b_X} \\
	&\Leftrightarrow \left( 1 + e^{\frac{1}{T_y}(-(a_Y+b_Y)x+b_Y)} \right)^{-1} < \frac{b_X}{a_X+b_X}\\
	&\Leftrightarrow x < \min\left\{\frac{1}{2},\frac{-T_y\ln\left(\frac{a_X}{b_X}\right)+b_Y}{a_Y+b_Y}\right\}
	\end{align*}
	Note that for $T_y>\frac{b_Y}{\ln(a_X/b_X)}=T_B$, we have $x_1<0$, and hence $T_X<0$.
	
	From the above derivation we can see that for all $x \in (0,1/2)$ such that $T_X^{II}(x,T_y)>0$, we have $y^{II} < 1/2$ since $\frac{b_X}{a_X+b_X}<1/2$.  Then, we can easily find that
	$$
	\frac{\partial L}{\partial x} = \left[(1-2x)+x(1-x)(1-2y^{II})\frac{a_Y+b_Y}{T_y}\right]\ln\left(\frac{1}{x}-1\right)\frac{\partial y^{II}}{\partial x} > 0.
	$$
	Further, when $T_y<T_I$, we have $y^{II}(1/2,T_y)<\frac{b_X}{a_X+b_X}$.  This implies that for $x \in (0,1/2)$, $y^{II}(x,T_y)<\frac{b_X}{a_X+b_X}$.  Since $\frac{\partial L}{\partial x}>0$, and $L$ is continuous, we can see that $L(x,T_y)<\frac{b_X}{a_X+b_X}$ for $x \in (0,1/2)$.  This implies the fourth part of the proposition.
	
	Next, if we look at the derivative of $T_X^{II}$, 
	$$
	\frac{\partial T_X^{II}}{\partial x}(x,T_y)
	= \frac{-(a_X+b_X)L(x,T_y)+b_X}{x(1-x)[\ln(1/x-1)]^2}
	$$
	we can see that any critical point in $x \in (0,1/2)$ must satisfy $L(x,T_y)=\frac{b_X}{a_X+b_X}$.  When $T_y>T_I$, $x_1<1/2$, and we can see that $L(x_1, T_y) > y^{II}(x_1, T_y) = \frac{b_X}{a_X+b_X}$.  If $T_y<\frac{b_Y}{\ln(a_X/b_X)}$, then $\lim_{x\rightarrow 0+} T_X = y^{II}(0, T_Y) < \frac{b_X}{a_X+b_X}$.  Hence, there is exactly one critical point for $T_X$ for $x \in (0,x_1)$, which is a local maximum for $T_X$.  If $T_y>\frac{b_Y}{\ln(a_X/b_X)}$, then we can see that $T_X$ is always negative, in which case the critical temperature is zero.
\end{proof}

The results in Proposition~\ref{prop:case1a_left} not only applies for the case $a_Y \ge b_Y$ but also general cases about the characteristics on $(0,1/2)$.
According to this proposition, we can conclude the following things for the case $a_Y \ge b_Y$, as well as the case $a_Y<b_Y$ when $T_y>T_I$:
\begin{enumerate}
	\item The temperature $T_B=\frac{b_Y}{\ln(a_X/b_X)}$ determines whether there is a branch appears in $x \in (0,1/2)$.  
	\item There is some critical temperature $T_C$.  If we raise $T_x$ above $T_C$, then the system is always on the principal branch.
	\item The critical temperature $T_C$ is given as the solution to the equality $L(x,T_Y)=\frac{b_X}{a_X+b_X}$.
\end{enumerate}
When there is a positive critical temperature, though it has no closed form solution, we can perform binary search to look for $x \in (0, x_1)$ that satisfies $L(x,T_y)=\frac{b_X}{a_X+b_X}$.


Another result we are able to obtain from Proposition~\ref{prop:case1a_left} is that the principal branch for Case~1a when $T_y<T_I$ lies on $(0,1/2)$.  

\begin{proof}[Proof of Part 2 of Theorem~\ref{thm:coord_topo_pb}]
	First, we note that $T_y<T_I$ is meaningful only when $b_Y>a_Y$, for which case we always have $T_I<T_B$.
	From Proposition~\ref{prop:case1a_left}, we can see that for $T_Y^{II}<T_I$, we have $x_1=1/2$, and hence $T_X^{II}>0$ for $x\in(0,1/2)$.  From Proposition~\ref{prop:dir_pb}, we already have $\lim_{x \rightarrow \frac{1}{2}^-} T_X^{II} = \infty$.  Also, it is easy to see that $\lim_{x \rightarrow 0^+} T_X^{II} = 0$.  As a result, since $T_X^{II}$ is continuous differentiable over $(0, 0.5)$, for any $T_x>0$, there exists $x \in (0,0.5)$ such that $T_X^{II}(x,T_y)=T_x$.
\end{proof}


What remains to show is the characteristics on the side $(1/2,1)$ when $b_Y>a_Y$.  In Figure~\ref{fig:as_c1a_by_ex1} and Figure~\ref{fig:as_c1a_by_ex2}, we can find that for low $T_y$, the branch on the side $(1/2,1)$ demonstrated a similar behavior as what we have shown in Proposition~\ref{prop:case1a_left} for the side $(0,1/2)$.  However, for high $T_y$, while we still can find that $(0,1/2)$ contains the principal branch, the principal branch is not continuous.  These observations are formalized in the following proposition.  From this proposition, the proof of part~4 of Theorem~\ref{thm:coord_topo_ct} directly follows.

\begin{proposition} \label{prop:c1a_g_right2}
	Consider Case~1a with $b_Y>a_Y$.  Let $x_2=\max\left\{\frac{1}{2},\frac{-T_Y\ln\left(\frac{a_X}{b_X}\right)+b_Y}{a_Y+b_Y}\right\}$ and $T_A=\max\left\{0, \frac{-a_Y}{\ln(a_X/b_X)}\right\}$.  The following statements are true for $x\in(1/2,1)$.
	\begin{enumerate}
		\item If $T_y<T_A$, then $T_X^{II}<0$.
		\item If $T_y>T_A$, then $T_X^{II}>0$ if and only if $x \in (x_2,1)$.
		\item For $x \in \left[\frac{b_Y}{a_Y+b_Y},1\right)$, we have $\frac{\partial L}{\partial x} > 0$.
		\item If $T_y \in (T_{A}, T_I)$, then there is a positive critical temperature $T_{C}(T_y)$ such that $T_X^{II}(x,T_y)\le T_{C}(T_y)$ for $x \in (1/2,1)$, given as $T_{C}(T_y)=T_X^{II}(x_L)$, where $x_L \in (1/2,1)$ is the unique solution of $L(x,T_y)=\frac{b_X}{a_X+b_X}$.
	\end{enumerate}
\end{proposition}
\begin{proof}
	For the first part and the second part, consider $x \in (1/2,1)$, and we can find that
	\begin{align*}
	T_X^{II} > 0 &\Leftrightarrow y^{II} > \frac{b_X}{a_X+b_X} \\
	&\Leftrightarrow \left( 1 + e^{\frac{1}{T_y}(-(a_Y+b_Y)x+b_Y)} \right)^{-1} > \frac{b_X}{a_X+b_X}\\
	&\Leftrightarrow x > \max\left\{\frac{1}{2},\frac{-T_y\ln\left(\frac{a_X}{b_X}\right)+b_Y}{a_Y+b_Y}\right\} = x_2
	\end{align*}
	Note that for $T_y>T_I$, we get $x_2=1/2$.  Also, if $T_y<T_{A}$, then $T_X^{II}<0$ for all $x \in (1/2,1)$.
	
	For the third part,  that $y^{II} \ge \frac{1}{2}$ for all $x\ge \frac{b_Y}{a_Y+b_Y}$ and $\frac{b_Y}{a_Y+b_Y}>\frac{1}{2}$.  Then, we can find that
	
	$$
	\frac{\partial L}{\partial x} = \left[(1-2x)+x(1-x)(1-2y^{II})\frac{a_Y+b_Y}{T_y}\right]\ln\left(\frac{1}{x}-1\right)\frac{\partial y^{II}}{\partial x} > 0
	$$
	
	For the fourth part, we can find that any critical point of $L(x,T_Y)$ in $(0,1)$ must be either $x=\frac{1}{2}$ or satisfies the following equation:
	\begin{equation} \label{eq:G_eq_0}
	(1-2x)+x(1-x)(1-2y^{II})\frac{a_Y+b_Y}{T_y} = 0
	\end{equation}
	Consider $G(x,T_y) = (1-2x)+x(1-x)(1-2y^{II})\frac{a_Y+b_Y}{T_y}$.  For $b_Y>a_Y$, $y^{II}(1/2,T_y)$ is strictly less than $1/2$.  Also, we can see that $\frac{b_Y}{a_Y+b_Y}>1/2$.  Now, we can observe that  $G(1/2,T_y)>0$ and $G(\frac{b_Y}{a_Y+b_Y},T_y)<0$.  Next, we can see that $G(x,T_y)$ is monotonic decreasing with respect to $x$ for $x \in \left(\frac{1}{2}, \frac{b_Y}{a_Y+b_Y}\right)$ by looking at its derivative:
	$$
	\frac{\partial G(x,T_y)}{\partial x} = -2+ \frac{a_Y+b_Y}{T_y}\left[ (1-2x)(1-2y^{II}) -2 x(1-x)\frac{\partial y^{II}}{\partial x} \right] < 0
	$$
	As a result, we can see that there is some $x^* \in \left(\frac{1}{2}, \frac{b_Y}{a_Y+b_Y}\right)$ such that $G(x^*,T_y)=0$.  This implies that $L(x,T_y)$ has exactly one critical point for $x \in \left(\frac{1}{2}, \frac{b_Y}{a_Y+b_Y}\right)$.  Besides, we can see that if $G(x,T_y)>0$, $\frac{\partial L}{\partial x}<0$; while if $G(x,T_y)<0$, then $\frac{\partial L}{\partial x}>0$.  Therefore, $x^*$ is a local minimum for $L$.
	
	From the above arguments, we can conclude that the shape of $L(x,T_y)$ for $T_y<T_I$ is as follows:
	\begin{enumerate}
		\item There is a local maximum at $x=1/2$, where $L(1/2,T_y)=y(1/2,T_y)<\frac{b_X}{a_X+b_X}$.
		\item $L$ is decreasing on the interval $\left(\frac{1}{2}, x^*\right)$, where $x^*$ is the unique solution to \eqref{eq:G_eq_0}.  
		\item $L$ is increasing on the interval $(x^*,1)$. If $T_y>T_{A}$, then $\lim_{x\rightarrow 1^-} L(x,T_y)=y(1,T_y) >\frac{b_X}{a_X+b_X}$.
	\end{enumerate}
	Finally, we can claim that there is a unique solution to $L(x,T_Y)=\frac{b_X}{a_X+b_X}$, and such point gives a local maximum to $T_X^{II}$.
\end{proof}

The above proposition suggests that for $T_y \in (T_{A},T_I)$, we are able to use binary search to find the critical temperature.  For $T_y>T_I$, unfortunately, with the similar argument of Proposition~\ref{prop:c1a_g_right2}, we can find that there are potentially at most two critical points for $T_X^{II}$ on $(1/2,1)$, as shown in Figure~\ref{fig:as_c1a_by_ex2}, which may induce an unstable segment between two stable segments.  This also proves part 3 of Theorem~\ref{thm:coord_topo_stab}.

Now, we have enough materials to prove the remaining statements in Section~\ref{sec:coord_topo}.
\begin{proof}[Proof of Part~1, 5, and 6 of Theorem~\ref{thm:coord_topo_ct}, part~4 of Theorem~\ref{thm:coord_topo_stab}]
	For $T_y>T_I$, by Proposition~\ref{prop:case1a_left}, we can conclude that for $x \in (0,x_L)$, we have $\frac{\partial T_X^{II}}{\partial x}>0$, for which the QREs are stable by Lemma~\ref{lem:qre_stab}.  With similar argument we can conclude that the QREs on $x \in (x_L,x_1)$ are unstable.  Besides, given $T_x$, the stable QRE $x_a\in(0,x_L)$ and the unstable $x_b\in(x_L,x_1)$ that satisfies $T_X^{II}(x_a,T_y)=T_X^{II}(x_b,T_y)=T_x$ appear in pairs.  For $T_y<T_I$, with the same technique and by Proposition~\ref{prop:c1a_g_right2}, we can claim that the QREs in $x \in (x_2, x_L)$ are unstable; while the QREs in $x \in (x_L,1)$ are stable.  This proves the first part of of Theorem~\ref{thm:coord_topo_ct} and part~4 of Theorem~\ref{thm:coord_topo_stab}.
	
	Part 5 and 6 of Theorem~\ref{thm:coord_topo_ct} are corollaries of part 5 of Proposition~\ref{prop:case1a_left} and part 4 of  Proposition~\ref{prop:c1a_g_right2}.
\end{proof}

\subsubsection{Case 1b: $b_X>0$, $a_Y+b_Y<0$} \label{sec:case1b}

In this case, both player have different preferences. For the game within this class, there is only one Nash equilibrium (either pure or mixed).  We presented examples in Figure~\ref{fig:as_c1b_low_TY} and Figure~\ref{fig:as_c1b_high_TY}.  We can find that in these figures, there is only one QRE given $T_x$ and $T_y$.  We show in the following two propositions that this observation is true for all instances.
\begin{proposition}
	Consider Case~1b. Let $x_3=\max\left\{0,\frac{-T_y \ln(a_X/b_X)+b_Y}{a_Y+b_Y}\right\}$. If $T_y<T_I$, then the following statements are true
	\begin{enumerate}
		\item $T_X^{II}(x,T_y)<0$ for $x \in (1/2,1)$.
		\item $T_X^{II}(x,T_y)>0$ for $x \in \left( x_3, \frac{1}{2} \right)$.
		\item $\frac{\partial T_X^{II}(x,T_y)}{\partial x}>0$ for $x \in \left( x_3, \frac{1}{2} \right)$.
		\item $\left( x_3, \frac{1}{2} \right)$ contains the principal branch.
	\end{enumerate}
\end{proposition}
\begin{proof}
	Note that if $T_y<T_I$, we have $x_3<1/2$. Also, according to Proposition~\ref{prop:c1a_g_monotone}, $y^{II}(1/2,T_y) < \frac{b_X}{a_X+b_X}$.  Since $y^{II}$ is continuous and monotonic decreasing with $x$, we can see that $y^{II} < \frac{b_X}{a_X+b_X}$ for $x>1/2$.  Therefore, the numerator of \eqref{eq:Tx_} is always positive for $x \in (1/2,1)$, which makes $T_X^{II}$ negative.  This proves the first part of the proposition.
	
	For the second part, observe that for $x \in (0, 1/2)$, $T_X^{II}>0$ if and only if $y^{II} < \frac{b_X}{a_X+b_X}$.  This is equivalent to $x > \frac{-T_y \ln(a_X/b_X)+b_Y}{a_Y+b_Y}$.
	
	For the third part, note that for $x \in (0,1/2)$, $x(1-x)\ln(1/x-1)\frac{\partial y^{II}}{\partial x} < 0$.  This implies $L(x,T_y) < y^{II}(x,T_y) < \frac{b_X}{a_X+b_X}$ for $x \in (x_3,1/2)$, from which we can conclude that $\frac{\partial T_X^{II}(x,T_y)}{\partial x}>0$.
	
	Finally, we note that if $x_3>0$, then $T_X^{II}(x_3,T_y)=0$.  If $x_3=0$, we have $\lim_{x\rightarrow 0^+} T_X^{II}=0$. As a result, we can conclude that $(x_3,1/2)$ contains the principal branch.
\end{proof}

With the similar arguments, we are able to show the following proposition for $T_y>T_I$:

\begin{proposition}
	Consider Case~1b. Let $x_3=\min\left\{1,\frac{-T_y \ln(a_X/b_X)+b_Y}{a_Y+b_Y}\right\}$. If $T_y>T_I$, then the following statements are true
	\begin{enumerate}
		\item $T_X^{II}(x,T_y)<0$ for $x \in (0,1/2)$.
		\item $T_X^{II}(x,T_y)>0$ for $x \in \left( \frac{1}{2}, x_3 \right)$.
		\item $\frac{\partial T_X^{II}(x,T_y)}{\partial x}<0$ for $x \in \left( \frac{1}{2}, x_3 \right)$.
		\item $\left( \frac{1}{2}, x_3 \right)$ contains the principal branch.
	\end{enumerate}
\end{proposition}
%
%

\subsubsection{Case 1c: $a_Y+b+Y=0$}
In this case, we have $T_I=\frac{b_Y}{\ln(a_X/b_X)}$, and $y^{II}$ is a constant with respect to $x$.  The proof of Theorem~\ref{thm:nc2} for $a_Y+b_Y=0$ directly follows from the following proposition.
\begin{proposition}
	Consider Case~1c.  The following statements are true:
	\begin{enumerate}
		\item If $T_y<T_I$, then $T_X^{II}(x,T_y)<0$ for $x \in (0.5,1)$, and $T_X^{II}(x,T_y)>0$ for $x \in (0,0.5)$.
		\item If $T_y>T_I$, then $T_X^{II}(x,T_y)<0$ for $x \in (0,0.5)$, and $T_X^{II}(x,T_y)>0$ for $x \in (0.5,1)$.
		\item If $T_y<T_I$, then $\frac{\partial T_X^{II}(x,T_y)}{\partial x}>0$ for $x \in \left( 0, 0.5 \right)$.
		\item If $T_y>T_I$, then $\frac{\partial T_X^{II}(x,T_y)}{\partial x}<0$ for $x \in \left( 0.5,1 \right)$.
	\end{enumerate}
\end{proposition}
\begin{proof}
	Note that $y^{II}=\left(1+e^{b_Y/T_y}\right)^{-1}$.  
	
	First consider the case when $a_Y>b_Y$.  In this case $T_I=0$ and $b_Y<0$.  Therefore, $y^{II}>\frac{b_X}{a_X+b_X}$, and from which we can conclude that $T_X^{II}>0$ for $x \in (0.5,1)$ and $T_X^{II}<0$ for $x \in (0,0.5)$, for any positive $T_y$.
	
	Now consider the case that $a_Y<b_Y$.  If $T_y<T_I$, we have $y^{II}<\frac{b_X}{a_X+b_X}$, and hence we get $T_X^{II}(x,T_y)<0$ for $x \in (0.5,1)$, and $T_X^{II}(x,T_y)>0$ for $x \in (0,0.5)$, which is the first part of the proposition statement.  Similarly, if $T_y>T_I$, we have $y^{II}>\frac{b_X}{a_X+b_X}$, from which the second part of the proposition follows.
	
	For the third part and the fourth part, note that $L(x,T_y)=y^{II}$ in this case as $\frac{\partial y^{II}}{\partial x}=0$ by observing \eqref{eq:def_L}, and the sign of the derivative of $T_X^{II}$ can be seen from Lemma~\ref{lem:techs}.
\end{proof}

\subsection{Case 2: $b_X<0$}\label{sec:case2}


In this case, the first action is a dominating strategy for the first player. Note that both $-(a_X+b_X)$ and $b_X$ are not positive, which means that the numerator of \eqref{eq:Tx_} is always smaller than or equal to zero.  This implies that all QRE correspondences appear on $x \in \left(\frac{1}{2}, 1\right)$.  In fact, since $y^{II}>0$ for $x \in (1/2,1)$, the numerator of \eqref{eq:Tx_} is always negative, we have $T_X^{II}>0$ for $x \in (1/2,1)$.  Also we can easily see that
$$
\lim_{x\rightarrow \frac{1}{2}^+} T_X^{II}(x,T_y) = +\infty
$$

This implies that $(1/2,1)$ contains the principal branch.  First, we show the result when $a_Y+b_Y<0$ in the following proposition.  Also, the bifurcation diagram is presented in Figure~\ref{fig:as_c2_ex1}.

\begin{proposition}
	For Case~2, if $a_Y+b_Y<0$, then for $x \in (1/2,1)$, we have $\frac{\partial T_X^{II}}{\partial x}<0$.
\end{proposition}
\begin{proof}
	In this case, $y^{II}$ is monotonic decreasing with $x$. We can see that
	$$
	L(x,T_Y) = y^{II} + x(1-x)\ln\left(\frac{1}{x}-1\right) \frac{\partial y^{II}}{\partial x} > y^{II} >0
	$$
	since $x(1-x)\ln\left(\frac{1}{x}-1\right) \frac{\partial y^{II}}{\partial x}$ is positive for $x \in (1/2,1)$.  Bringing this back to \eqref{eq:def_dTX}, we have $\frac{\partial T_X^{II}}{\partial x}<0$.
\end{proof}

For $a_Y+b_Y>0$, if $a_Y>b_Y$, the bifurcation diagram has the similar trend as in Figure~\ref{fig:as_c2_ex1}; while if $a_Y<b_Y$, we lose the continuity on the principal branch.
\begin{proposition}
	For Case~2, if $a_Y+b_Y>0$, then for $x\in(1/2,1)$, we have
	\begin{enumerate}
		\item if $a_Y>b_Y$, then $\frac{\partial T_X^{II}}{\partial x}<0$.
		\item if $a_Y<b_Y$, then $T_X$ has at most two local extrema.
	\end{enumerate}
\end{proposition}
\begin{proof}
	In this case, $y^{II}$ is monotonic increasing with $x$.  For $a_Y>b_Y$, we can find that $y^{II}(1/2,T_y)>0$ and $L(1/2,T_y)=y^{II}(1/2,T_y)>0$. Also, we can get that $L$ is monotonic increasing for $x \in (1/2,1)$ by inspecting
	$$
	\frac{\partial L(x,T_y)}{\partial x} = \left[(1-2x)+x(1-x)(1-2y^{II})\frac{a_Y+b_Y}{T_y}\right]\ln\left(\frac{1}{x}-1\right) \frac{\partial y^{II}(x,T_y)}{\partial x} > 0
	$$
	Hence, for $x \in (1/2,1)$, $L(x,T_y)>0$.  This implies $\frac{\partial T_X^{II}}{\partial x}<0$ for $x\in(1/2,1)$.
	
	For the second part, we can find that for $a_Y<b_Y$, $y^{II}(1/2)<1/2$. Let $x_2 = \min\left\{1,\frac{b_Y}{a_Y+b_Y}\right\}$. First note that if $x_2<1$, then for $x>x_2$, we have $y>1/2$, and further we can get $\frac{\partial L(x,T_y)}{\partial x} >0$ for $x\in(x_2,1)$.  We use the same technique as in the proof of the Proposition~\ref{prop:c1a_g_right2}.  Let $G(x,T_y=(1-2x)+x(1-x)(1-2y^{II})\frac{a_Y+b_Y}{T_y}$.  Note that $G(1/2, T_y)>0$ and $G(x_2,T_y)<0$.  Next, observe that $G(x,T_y)$ is monotonic decreasing for $x \in \left(\frac{1}{2}, x_2\right)$.  Hence, there exists a $x^* \in (1/2, x_2)$ such that $G(x^*,T_y)=0$.  This $x^*$ is a local minimum for $L$.  We can conclude that for $x \in (1/2,1)$, $L$ has the following shape:
	\begin{enumerate}
		\item There is a local maximum at $x=1/2$, where $L(1/2,T_y)=y(1/2,T_y)>0$.
		\item $L$ is decreasing on the interval $x \in (1/2, x^*)$, where $x^*$ is the solution to $G(x^*,T_y)=0$.
		\item $L$ is increasing on the interval $x \in (x^*, x_2)$.  Note that $\lim_{x \rightarrow 1^-} L(x,T_y) = y^{II}(1,T_y)>0$.
	\end{enumerate}
	As a result, if $L(x^*, T_y)>\frac{b_X}{a_X+b_X}$, then $T_X^{II}$ is monotonic decreasing; otherwise, $T_X^{II}$ has a local minimum and a local maximum on $(1/2,1)$.
\end{proof}

\end{document}